\documentclass[10pt,journal]{IEEEtran}
\IEEEoverridecommandlockouts
\usepackage[dvips]{graphicx}
\usepackage{nomencl}
\makenomenclature
\usepackage{xspace}
\usepackage{amsmath}
\usepackage{caption}
\usepackage{subcaption}
\usepackage{placeins}
\usepackage{balance}
\usepackage{cite}
\usepackage{bm}
\usepackage{amsmath}
\usepackage{array}
\usepackage{setspace} 
\usepackage[nameinlink]{cleveref}
\usepackage{accents}
\usepackage{mathtools}
\usepackage{algpseudocode}

\usepackage{amsthm} 
\usepackage{multirow}
\usepackage[linesnumbered,ruled,vlined]{algorithm2e}
\usepackage{tabularx}
\SetKwRepeat{Do}{do}{while}
\usepackage{color}
\usepackage{epstopdf}
\usepackage{endnotes}
\usepackage{framed}

\usepackage{cool} 
\usepackage{enumerate} 

\usepackage{pifont}

\Crefname{figure}{Fig.}{Figs.}

\usepackage{mathrsfs}
\usepackage[justification=centering]{caption}
\allowdisplaybreaks[4]

\newcolumntype{P}[1]{>{\centering\arraybackslash}p{#1}}

\DeclareUnicodeCharacter{2212}{-}

\newtheorem{definition}{Definition}

\newtheorem{remark}{Remark}

\newtheorem{lemma}{Lemma}

\usepackage{color}
\usepackage{changes}
\usepackage{lipsum}
\definechangesauthor[name={Abdullatif}, color=red]{revised/added}
\definechangesauthor[name={Abdullatif}, color=red]{Revised/Added}
\definechangesauthor[name={Abdullatif2}, color=red]{deleted}
\definechangesauthor[name={Abdullatif3}, color=red]{corrected}
\definechangesauthor[name={Abdullatif4}, color=red]{moved}
    
\def\BibTeX{{\rm B\kern-.05em{\sc i\kern-.025em b}\kern-.08em
        T\kern-.1667em\lower.7ex\hbox{E}\kern-.125emX}}
\begin{document}

\title{FedPot: A Quality-Aware Collaborative and Incentivized  Honeypot-Based Detector for Smart Grid Networks}

\author{Abdullatif Albaseer,~\IEEEmembership{Member,~IEEE,}
        Nima Abdi, 
        Mohamed Abdallah,~\IEEEmembership{Senior Member,~IEEE,} Marwa Qaraqe  ~\IEEEmembership{Senior Member,~IEEE,} Saif Al-Kuwari,~\IEEEmembership{Senior Member,~IEEE}
\thanks{Abdullatif Albaseer,Nima Abdi, Mohamed Abdallah, Marwa Qaraqe, and  Saif Alkuwari are with Division of Information and Computing Technology, College of Science and Engineering, Hamad Bin Khalifa University, Doha, Qatar (e-mail:\{aalbaseer, niab52126, moabdallah, mqaraqe, smalkuwari\}@hbku.edu.qa).}
\thanks{* Preliminary results in this work are presented at the IEEE CCNC Conference, 2023~\cite{10060393}}
}

\maketitle

\begin{abstract}
Honeypot technologies provide an effective defense strategy for the Industrial Internet of Things (IIoT), particularly {in} enhancing the Advanced Metering Infrastructure's (AMI) security by bolstering the network intrusion detection system. For this security paradigm to be fully realized, it necessitates the active participation of small-scale power suppliers (SPSs) in implementing honeypots and engaging in collaborative data sharing with traditional power retailers (TPRs). To motivate this interaction, TPRs incentivize data sharing with tangible rewards. However, without access to {an} SPS's confidential data, it is daunting for TPRs to validate shared data, {thereby} risking SPSs' privacy and increasing sharing costs due to voluminous honeypot logs. These challenges can be resolved by utilizing Federated Learning (FL), a distributed machine learning (ML) technique {that} allows for model training without data relocation. However, the conventional FL algorithm lacks the requisite functionality for {both} the security defense model and the rewards system of the AMI network. This work presents two solutions: {first,} an enhanced {and} cost-efficient FedAvg algorithm incorporating a novel data quality measure, and {second,} FedPot, the development of an effective security model with {a} fair incentives mechanism under {an} FL architecture. Accordingly, SPSs are limited to sharing the ML model they learn after efficiently measuring their local data quality, whereas TPRs can verify the participants' uploaded models and fairly compensate each participant for their contributions through rewards. Moreover, the proposed scheme addresses the problem of harmful participants who share subpar models while claiming high-quality data through a two-step verification approach. {Simulation} results, drawn from realistic {mircorgrid} network log datasets, demonstrate that the proposed solutions outperform state-of-the-art techniques by enhancing the security model and guaranteeing fair reward distributions.
\end{abstract}


%
\IEEEpeerreviewmaketitle

{
\printnomenclature
\nomenclature{SG}{Smart Grid}
\nomenclature{IIoT}{Industrial Internet of Things}
\nomenclature{SPSs}{Small-Scale Power Suppliers}
\nomenclature{TPRs}{Traditional power retailers}
\nomenclature{AMI}{Advanced Metering Infrastructure}
\nomenclature{DoS}{Denial-of-Service}
\nomenclature{DDoS}{Distributed Denial-of-Service}
\nomenclature{NIDS}{Network Intrusion Detection System}
\nomenclature{ML}{Machine Learning}
\nomenclature{DL}{Deep Learning}
\nomenclature{GDPR}{General Data Protection Regulation}
\nomenclature{FL}{Federated Learning}
\nomenclature{VDD}{Valid Defense Data}
\nomenclature{CPU}{Central Processing Unit}
\nomenclature{IR}{Individual Rationality}
\nomenclature{IC}{Incentive Compatibility}
\nomenclature{LDIC}{Local Downward Incentive Compatibility}
\nomenclature{LUIC}{Local Upward Incentive Compatibility}
\nomenclature{PLCs}{Programmable Logic Controllers}
\nomenclature{RTUs}{Remote Terminal Units}
\nomenclature{TPRate}{True Positive Rate}
\nomenclature{TNR}{True Negative Rate}
\nomenclature{MITM}{Man-in-the-Middle}
\nomenclature{IID}{Independent and Identically Distributed}
\nomenclature{$M$}{number of connected SPSs}
\nomenclature{$\theta_m$}{SPS type: $\theta_1\le \ldots \le \theta_m \le \ldots \le \theta_M$}
\nomenclature{$\boldsymbol{V}_m$}{local model of SPS $m$}
\nomenclature{${\Pi}$}{Incentive set for all types}
\nomenclature{$\pi_m$}{the model and its rewards of type $m$}
\nomenclature{$\boldsymbol{V}(z)$}{global model at $z$-\textit{th} training round}
\nomenclature{$R_m$}{the reward given to SPS $m$}
\nomenclature{$\boldsymbol{V}_m(j)$}{the local model update at $j$ local iteration}
\nomenclature{$F_m(\boldsymbol{V}_{m}(j))$}{the local loss function at $j$ local iteration}
\nomenclature{$F_m(\boldsymbol{V}_{m})$}{the updated local loss function}
\nomenclature{${F_{i}(\mathbf{\boldsymbol{V}_m}, x_{i}, y_{i}})$}{loss function on sample $i$}
\nomenclature{$E$}{number of local training epochs}
\nomenclature{$\eta$}{the learning rate}
\nomenclature{$\mathcal{S}_z$}{the selection set at round $z$}
\nomenclature{$S^m_z$}{the selection binary variable}
\nomenclature{$\xi$}{the defense model size}
\nomenclature{$C_d$}{the honeypot deploying cost}
\nomenclature{$C_u$}{the uploading cost of the local defense model}
\nomenclature{$C_t$}{the local training cost to update the defense model}
\nomenclature{$C_m$}{the local deployment, training, and uploading costs}
\nomenclature{$T^{max}$}{the maximum latency in each round}
\nomenclature{$U_m$}{the utility of SPS $m$}
\nomenclature{$U_{TPR}$}{the utility of TPR (i.e., utility company)}
\nomenclature{$G(V_{m})$}{the revenue given by each SPS $m$}
\nomenclature{$\Lambda(x_{i}, \delta)$}{an open ball space with a radius of $\delta$ centered at $x_i$}
\nomenclature{$\varphi(\mathcal{D}_m)$}{the estimation of the local VDD quality}}

\section{Introduction}
%
%
%
%

{ \IEEEPARstart{A} Smart Grid (SG) is an advanced electricity system that leverages digital tech, IIoT, and networking to boost efficiency, reliability, and sustainability. It achieves this through a cyber-physical system for the bidirectional flow of power and information, automating supply-demand balance with real-time data \cite{sg1}. SG allows individuals to become small-scale power suppliers (SPSs) using renewable sources, reducing the burden on traditional power retailers (TPRs) and fostering advanced metering infrastructure (AMI) integration \cite{3,7725484}.}

However, integrating different network technologies to interconnect SPSs and TPRs, which lack a proper defense system, is raising concerns about the security and privacy of SG systems.  Adversaries have access to numerous vulnerabilities, including ways to launch destructive attacks and access sensitive information such as a homeowner's residential address, social security number, daily habits, and any information related to unauthorized electricity consumption or disruption in the network. {As an example of such attacks, the 2015 Ukrainian power outage illustrates the critical vulnerabilities to cyber-attacks in both the control center and the smart devices employed for managing and observing the electrical system \cite{7752958}.} 

In addition, SG is subject to denial-of-service (DoS) attacks, which can flood the network with traffic{,} causing delays in data transmission and processing. This can lead to disruptions in the system's normal operation and potentially cause critical elements of the energy system to fail~\cite{sg3}.
To protect the communication infrastructure, {implementing} a Network Intrusion Detection System (NIDS) is crucial. NIDS serves as a robust shield, detecting numerous threats that extend beyond the capacity of traditional firewalls. {The} protective capability of NIDS has been notably enhanced by the advent of Machine Learning (ML) and Deep Learning (DL)-based approaches, pushing their performance to a significantly higher level. This interconnected relationship between NIDS, ML, and DL demonstrates a synergistic blend of network security and advanced computational methodologies~\cite{10399957}.

Building upon this security paradigm,  {incorporating}   honeypot technology further enriches the protective capabilities of ML/DL-based NIDSs. Honeypots empower these systems to meticulously map the attack surface, discern patterns, and thwart malicious actions by delivering detailed insights into potential intruder behavior{.}
Within the realm of SG, a honeypot mimics the regular operations of a meter with the intent to deceive, misdirect, and analyze the activities of potential intruders. Employing such tactics allows SPSs to develop independent, streamlined security measures. Furthermore, it facilitates the exchange of defensive information with TPRs, {eliminating} the need for TPRs to purchase costly security models from security retailers.
This approach amplifies the protective layers of AMI, creating a more robust defense mechanism. Simultaneously, it alleviates the financial strain associated with defense strategies for TPRs, presenting a dual advantage {by} strengthening network security while managing costs effectively~\cite{9}.

Significant effort in the literature has been devoted to developing incentive mechanisms and designing contracts to encourage SPSs to implement honeypots and share protective information with TPRs while maintaining a balance such that SPSs do not reap excessive benefits beyond their due. For example, in ~\cite{13,15}, the authors proposed information asymmetry-based contract theory approaches considering different communication systems. Concentrating on the SG network, the work in~\cite{9397770} introduced a recent motivation-driven approach where TPR motivates SPSs to deploy honeypots and exchange defensive information to enhance the security framework. This approach relies on a range of essential prerequisites, including the dimensions of the shared data and the related expenses.
However, all the prior work may violate the SPSs’ privacy
while increasing the cost of sharing the obtained raw data. 
The recent enforcement of stringent data privacy regulations, such as the General Data Protection Regulation (GDPR) \cite{dataprivacyref}, further underscores the importance of maintaining data privacy.
Additionally, the large volume of honeypot logs can lead to excessive transmission costs and network congestion. It is also worth noting that the shared honeypot logs may not always yield benefits, as they could contain redundant information that does not enhance the existing protection strategy.
Given the aforementioned concerns, federated learning (FL)~\cite{9415623} has become an effective distributed ML/DL method to maintain privacy and minimize communication costs by exclusively sharing the ML model while retaining the raw data in its original repository. This seamlessly connects the need for privacy preservation and cost-efficiency with the advantages of shared data for network security.
Current research on FL typically assumes that entities are willing to participate in the FL training process and use the collected data to improve the shared model~\cite{incentive1,incentive3}. {In reality, entities may be reluctant to join without properly designed incentives (i.e., contracts)} because FL consumes significant resources (i.e., computation and communication costs). Furthermore, entities in FL are autonomous actors who decide when and how to interact. When dealing with different reward strategies from different alliances, participants may use different training techniques, influencing the performance of the resulting models. To that end, it is critical to develop an efficient {reward} mechanism to encourage entities to participate in FL while maintaining the required level of data quality. 
Researchers have recently presented several incentive schemes aimed at compensating the involved parties (i.e., SPSs) using financial incentives according to the magnitude of their data contributions \cite{incentive12, incentive15, incentive18}. Utilizing current FL methodologies, all participants acquire the same model upon the end of the training phase, regardless of the produced data quality (i.e., honeypot logs) or the impact of the submitted local models. Consequently, certain participants possessing large datasets may provide inadequate contributions yet attain a higher {proficiency level than} others who possess high-quality data. This scenario gives rise to a challenging issue called the FL free-rider problem.

\subsection{Contribution}
Considering the previously mentioned remarks, {we introduce FedPot, an FL framework initially designed for SG networks but validated across diverse scenarios, including IEC 104 and N-BaIoT datasets. This architecture incorporates refined, efficient, and resilient aggregation and averaging techniques complemented by a fair rewards system based on data quality. While our primary case study is centered on SGs, {our approach of combining honeypot logs and FL has shown adaptability and effectiveness across multiple domains.}} 
We introduce novel schemes for local data quality, participant selection, and global model averaging, where the TPR resolves a convex optimization problem that prioritizes data quality over data size. Each SPS fine-tunes the global model received from the TPR using its honeypot logs and transmits back the model parameters. After that, the TPR combines and enhances the defensive model employing the proposed approaches, as detailed later. To overcome the FL free-rider problem within AMI networks, we propose a novel metric to measure local data quality and contributions instead of data size, which may contain redundant information that does not improve the security defense model. We devised a two-step verification process to tackle malicious or poorly performing participants. The TPR verifies submitted models, then updates the global model and allocates incentives based on contribution.
In summary, the primary contributions of this work are:
\begin{itemize}
    \item Develop an effective architecture for privacy-preserving honeypot-based detectors, FedPot, that protects user privacy while considering data quality, an efficient global model, and fair incentives. The proposed solution handles and ensures a reliable FL training process based on valid defense data acquired by implementing honeypots on the SPS side.
    \item Formulate the problem as an optimization problem, then provide solutions incorporating (i) problem reformulation and transformation, (ii) the prior quality determination of the local data through novel metrics, and (iii) two schemes for reliable global averaging and contribution-based reward distributions. 
    \item  {In response to the challenge of adversarial perturbations in model uploads, we introduce a universally applicable two-step verification system. This robust approach is designed to ensure the integrity of model contributions.}
    \item  {Carry out comprehensive simulations using realistic log data from various datasets (i.e., N-BaIoT, IEC 104, and IEC MMS datasets). The results affirm that our proposed framework outperforms existing state-of-the-art techniques in multiple application domains.}
\end{itemize}
\subsection{Organization}
The remainder of the paper is structured as follows. 
Section \ref{Sec:relatedwork} briefly reviews the relevant literature on NIDSs and honeypot deployment. 
The system models, including the learning, cost, and reward, are presented in \Cref{sec:system_model}. The problem is formulated in \Cref{sec:problemformulation}, and the proposed solution, including the problem reformulation and transformation, is presented in \Cref{sec:proposedsol}. In \Cref{sec:eval}, we discuss the experimental setup and present the numerical results. Finally, \Cref{sec:conclusion} concludes this paper and suggests possible directions for further research.

\section{Related Work}
\label{Sec:relatedwork}
With the proliferation of microgrids, major concerns about cyberattacks on such systems via smart meters have arisen. The United States Department of Homeland Security reported 224 destructive cyber intrusions against local electric utilities between 2013 and 2014 \cite{5}.
\paragraph{AMI Security}
Security concerns have been intensively investigated in the past few years as a key component of the IIoT. Authors in \cite{du2019sdn} studied the security approach to mitigate cyber-attacks in the context of the IIoT. The main assumption was that the attackers have sufficient tools to identify advanced vulnerabilities that enable them to attack IIoT systems. In \cite{li2017consortium}, Li \emph{et al.} used consortium blockchain technology to overcome transaction constraints in the IIoT. However, fewer studies have been conducted on the security of various components of AMI systems. The work in \cite{7036881} introduced a security protocol that preserves AMI private information while securely delivering control packets at the exact time. In \cite{6720175}, Fasial \emph{et al.} investigated the feasibility of employing data stream mining to improve AMI cybersecurity via NIDS. Yan \emph{et al.} \cite{6574667} describe an SG AMI security framework where different security concerns associated with AMI deployment are considered. 

\paragraph{ML/DL-based and FL-based NIDSs} 
 Many ML algorithms have been utilized to boost the NIDS in the past years to understand complicated network traffic better \cite{10328057,eddin2022fine,albaseer2022fine}. ML/DL-based NIDSs are used to identify unknown intrusions by analyzing the statistical characteristics of the network traffic. However, DL-based solutions have shown better performance, especially in extracting knowledge from complicated features rather than the shallow features in ML-based detectors \cite{10328057,eddin2022fine,albaseer2022fine}. {Recently, FL has been increasingly employed in the realm of NIDSs for the collaborative design of ML and DL-based detectors \cite{kang2019incentive}. Specifically, the work presented in \cite{popoola2021federated} developed a cooperative detector capable of identifying zero-day botnet attacks in oT networks using FL. Extending FL's utility to energy systems, \cite{badr2023privacy} introduced a privacy-preserving and communication-efficient FL-based energy predictor for net-metering systems. This approach combines a hybrid DL model for high-accuracy energy forecasting with an Inner-Product Functional Encryption scheme to encrypt local model parameters, maintaining data privacy.}

\paragraph{Honeypot Deployment Based Incentive Mechanisms}
Honeypots are practical security tools {that} deceive cyber attackers by acting as vulnerability traps \cite{20}. It has been widely used to enhance defensive performance on different systems \cite{9,21}. Tian \emph{et al.} \cite{21} developed a honeypot system to defend against APT attacks in the SG, mainly in the bus nodes. In their system, low-interaction and high-interaction honeypots were applied. Similarly, Wang \emph{et al.} \cite{7} proposed a honeypot scheme with various mixed distributions to detect the AMI network traffic. However, Wang \emph{et al.} failed to consider how using a honeypot system could reduce TPR defensive pressure while increasing AMI defensive effectiveness. Using a honeypot, Du and Wang \cite{9} captured distributed denial of service attacks (DDoS) attacks in AMI. The implemented methodology uses incomplete information static game and honeypot deployment to investigate non-cooperative problems. Householders, for example, install small wind turbines for small enterprises to sell excess electricity and get a profit from utility companies \cite{3}. 
The work in \cite{9397770} proposed a game theory-based approach to designing an incentive mechanism that allows SPSs to share their defense data with TPRs; thus, the TPR can pay the rewards accordingly. 

To summarize, significant effort has been made in the literature to overcome the challenge of deploying the honeypot and exchanging the obtained defensive data in AMI. 
However, most of these studies focused on log size and overlooked data redundancy that may be received from multiple SPSs. 
Furthermore, the privacy of the SPSs was completely ignored, posing an undesirable privacy risk to target consumers. 
Finally, there are additional costs associated with uploading large logs. 
Thus, in this paper, we propose FedPot. This efficient framework not only focuses on all these challenges but also solves advanced issues associated with malicious SPSs and ensures a more reliable defense model.

\section{System Model}
\label{sec:system_model}
\begin{figure}
	\centering  
	\includegraphics[width=\linewidth]{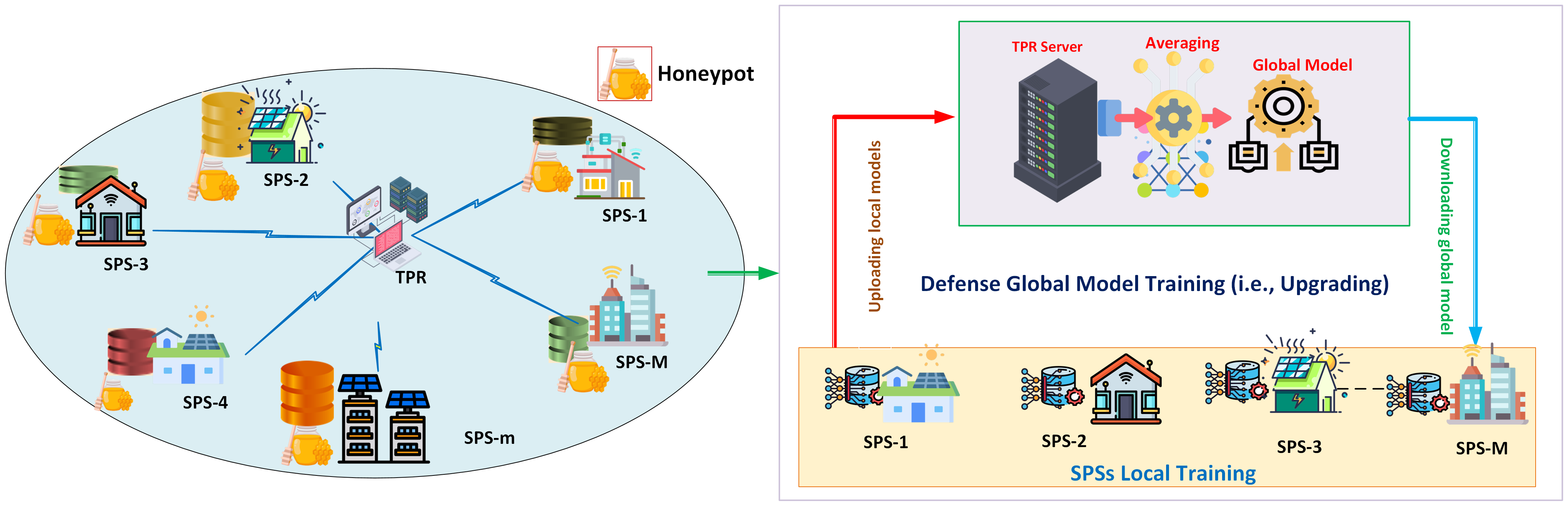}
	\caption{Micro-grid system where $M$ SPSs deploy honeypot and transfer security model with a single TPR.} 	\label{fig:sys_model}  
\end{figure}

This paper considers singular TPR within an AMI network, which gathers defensive data from $M$ interconnected SPSs, as depicted in \Cref{fig:sys_model}. 
{In our microgrid system, SPSs can acquire defense data by implementing and deploying honeypots. This includes mimicking real and different services (i.e., IEC 104 and GOOSE), applications, and typical vulnerabilities. The honeypot remains an attractive target while collecting diverse attack patterns to evolve with the changing landscape of threats. In this work, we consider two types of traffic captured by honeypots: fake traffic, which enables us to understand attacker tactics without data centralization, and legitimate traffic, which provides valuable insights into user behavior patterns processed in a decentralized manner to maintain privacy. To foster a diverse range of defense data, TPR may provide incentives to stimulate SPSs to establish honeypots, gather cyber-attack intelligence, encapsulate the defensive data, and transmit it to the TPR.}

Though SPSs can exchange defense data in the AMI network, there is an information imbalance between the TPR and SPS. This discrepancy arises because the SPS knows its valid defense data (VDD) while the TPR is not. The VDD includes the undisclosed attack event logs gathered by honeypots and used by the currently deployed TPR defensive system to improve the existing architecture.
{Practically, the honeypot logs amassed are extensive and could include users' sensitive information, leading to concerns surrounding computation and communication costs and potential privacy violations.}
FL can mitigate these issues by utilizing ML algorithms to extract knowledge from the local VDD of each SPS and share only the model attributes with the TPR.
Nevertheless, traditional FL Integration algorithms and incentive techniques are not ideally suited for the security defense framework within the AMI network. As a result, we introduce a privacy-preserving framework for effective honeypot-based detection. Our objective is to develop a method that promotes effective assurance of local model quality and provides incentives to bolster the defense system of the AMI network. We presume the SPSs hold $M$ distinct data types {that are separately distributed}.

Let $\theta_1\le \ldots \le \theta_m \le \ldots \le \theta_M$ denote the VDD type for each SPS, with a higher type signifying superior quality. To conduct FL model training, the TPR must ensure the desired quality at every global training round by devising an incentive set (i.e., contract) ${\Pi}=\{\pi_m=(\boldsymbol{V}_m, R_m)|m\in \{1,\ldots, M\}\}$, which establishes the relationship between rewards and local model quality based on the type, where $\boldsymbol{V}_m$ denotes the local model attributes of $m$ VDD type and $R_m$ represents the incentives provided by the TPR.
For participant selection, the TPR needs initial information such as how many SPSs are available to participate in defense model training, effort duration, declared data quality, and associated costs.
To guarantee an efficient training procedure, the TPR imposes a time constraint for every global round $z$, beyond which the model modifying and transferring must be finalized. The chosen SPSs employ their VDD to refine the global model and subsequently transfer the upgraded model to the TPR. According to \cite{incentive1}, the local model is formulated as follows:
\begin{equation}
\boldsymbol{V}_{m}(j)= \boldsymbol{V}_{m}(j-1)-\eta \nabla F_m(\boldsymbol{V}_{m}(j)),
\end{equation}
where $j = 1, 2, ...,$ $E\frac{D_m}{b}$ represents the index of local update for batch $b$ and the epoch count $E$, and
\begin{equation}
F_m(\mathbf{\boldsymbol{V}_m}) \triangleq \frac{1}{{D}_{m}}\sum_{i\in\mathcal{D}_{m}}{F_{i}(\mathbf{\boldsymbol{V}_m}, x_{i}, y_{i}}),
\end{equation}
denotes the local loss function, {quantifying }the model's error concerning local data samples, $D_m$ signifies the quantity of data examples in the gathered records (i.e., traffic logs), and $\eta$ stands for the learning rate. Upon completion of local training by all chosen SPSs, their updates are transmitted to the TPR. Then, the TPR consolidates these updates and modifies the global model given by:
\begin{equation}
{\mathbf{\boldsymbol{V}(z)}={\sum_{m=1}^{M}\frac{D_m}{D} \mathbf{\boldsymbol{V}}_{m}}\label{eq:globalAverage},}
\end{equation}
where
$
D = \sum_{m=1}^M D_m,
$
represents the aggregate samples across all SPSs.
Moreover, to guarantee that the attained accuracy is generalized, the TPR evaluates the uploaded models and provides incentives only to SPSs exceeding the threshold, $\psi_m$, as determined by the quality assessment algorithm discussed subsequently. SPSs that pass the quality evaluation receive rewards.
The TPR iterates this procedure for $Z$ global iterations until the shared security model reaches the target accuracy level.
In this paper, we highlight that the incentive set is constructed every global round, where the TPR aims to sustain the desired accuracy while allowing for a slight increment in succeeding rounds.
\subsection{Participation Willingness} 
In addition to the cost of deploying the honeypot, the suppliers bear the cost of training and uploading the local models. This cost includes the time and energy consumption for the honeypot deployment as well as training and uploading their local models. As such, the TPR cannot serve all existing devices due to insufficient bandwidth sub-channels. Thus, in every FL round, just a subset of the participants $\mathcal{S}_z$ {can} send their updates. The chosen set is described as follows:
\begin{equation}
    \mathcal{S}_z=\{m \mid S^m_z=1, m=1,2,...,M\},
\end{equation}
where $S^m_z=1$ denotes that device $m$ is in the selected set $\mathcal{S}_z$, otherwise $S^m_z=0$.
The latency among the chosen set consists of two parts: uploading latency and computing latency. 
For the uploading latency, all devices have the same defense model structure determined by TPR, with a size ${{\boldsymbol{v}}_z}$ denoted by $\xi $. Generally, we consider that the orthogonal frequency division multiple access (OFDMA) scheme is used for the communication between the BS and the devices, where each $m$-\textit{th} device is given bandwidth of size $\lambda^m_z B$. Additionally, based on the size of the model, the system bandwidth is split into N sub-channels  given by:
$
    N = \frac{B}{\xi},
$
where each participant is assigned an $1$-\textit{th} sub-channel of size $\lambda^m_z B$ . As a result, the m-\textit{th} device is capable of achieving a data rate of
$ \alpha_m^z=\lambda^m_z B\ln{\left( 1+\frac{P^{m}_z {|h^{m}_z|^2}}{N_0} \right)},
$ where $h^{m}_z$ represents the channel gain between the BS and client $m$, while $P^{m}_z$ denotes the transmit power of the $m$-\textit{th} client, and $N_0$ denotes background noise (i.e, Additive White Gaussian Noise (AWGN)). Hence, the uploading latency is given by:
$
    T_m^{trans}= \frac{\xi }{ \alpha_m^z}, \cdot  
$
and the related transmission cost is defined as:
$
    C_u = T_m^{trans} p_m^{trans},
$
where $p_m^{trans}$ is the transmit power. 
Regarding the calculated latency, every device requires a time duration defined as:
$
    T^m_{cmp}= E\frac{\phi_m D_m}{f_m},  
$
to train its local model, where $\phi_m$ (cycles/sample) denotes the number of processing cycles to execute one sample, and $f_m$ (cycles/second) is the central processing unit (CPU) frequency. Accordingly, the training cost is defined as: 
$
    C_t = \hat{\kappa}_mf_m^3 T^m_{cmp}, 
$
where ${\kappa}_m$ is the capacitor's coefficient related to the chip.
As a result, the combined latency for computation and uploading for each $m$-th participant during the $z$-th round is expressed as:
$
     T^m_{total}= T^m_{trans} + T^m_{cmp}.
$
In a real FL scenario, the coordinating server (i.e., TPR) establishes a time constraint, i.e., round deadline constraint, ensuring that every SPS participant completes their tasks within this specified timeframe.
Specifically, this time can be identified based on the latency of the slowest chosen SPS $m \in \mathcal{S}_z$, which is defined as:
\begin{equation} 
\label{eq:deadline}
T^{max} = \max\{S^m_z(T^m_{trans} + T^m_{cmp})\}.
\end{equation}
{It is worth emphasizing that while SPSs do benefit from the globally trained model, this alone does not ensure a model tailored to their unique threat landscape. Incentives encourage SPSs to actively contribute quality data, keeping the model updated and effective against evolving cybersecurity threats,{ leading} to solving the free-rider problem in FL. Incentives also offset the operational costs of maintaining honeypots and promote diverse participation. This results in a more robust and tailored global model, rewarding contributors with a system better suited to their security needs.}
\subsection{SPS and TPR Utilities}
To calculate the SPS utility, which involves implementing the honeypot and exchanging information associated with the VDD data (i.e., the updated local model), we initially have to compute the total incurred cost in the following manner:
\begin{equation}
  C_{m}  = C_d + C_t + C_u,
\end{equation}
where $C_d$ represents the cost related to running the honeypot, $C_t$ is the cost related to updating the local model utilizing the gathered logs, and $C_u$ is the cost related {to uploading} the local update. Accordingly, the utility of each SPS is expressed as:
\begin{equation}\label{SPS_utlity}
U_{m}=\ln(\theta_{m} R_{m}) -  C_{m}.
\end{equation}
Equ. \eqref{SPS_utlity} implies that while the reward increases, the utility does not increase at the same rate due to the diminishing returns property of the logarithmic function.
Once all SPSs taking part in the updating and uploading of the global defense model finish their tasks, the utility of the TPR is computed as follows:
\begin{equation}\label{Utility_server}
U_{TPR} = \sum_{m=1}^{M}( \theta_{m}G(V_{m})-R_{m}),
\end{equation}
{where} $G(V_{m})$ represents the revenue obtained from each SPS's update. Specifically, each SPS shares its model, $V_{m}$, which will add a profit, $G(V_{m})$ to TPR with associated costs.
\subsection{Adversarial Model and Assumption}
In this research paper, we analyze two stages of adversarial behavior using the N-BaIoT dataset. The first stage involves deploying botnets to carry out DDoS attacks by introducing networked zombies into the SPSs. {The SPS implements a honeypot system to counter these threats, collect valuable information about the attack surface,} and document all activities. These logs are then integrated into a collaborative model training approach, enabling information exchange with the TPR and enhancing understanding of the attacker's behavior. The second stage of adversarial conduct occurs during the FL process when a malicious SPS falsely claims superior local data quality while contributing subpar updates to the model. To develop a secure and efficient FL process that can mitigate such issues, we put forward two distinct aggregation and averaging techniques. These methods will be discussed in greater detail later in the paper, providing insight into their practical applications and advantages.
\section{Problem Formulation}
\label{sec:problemformulation}
The TPR aims to optimize data quality to boost the performance of the target model and assures fair incentives to all participating SPSs. This is equivalent to the optimization problem given by:
\begin{align}\label{Problem_1}
\textbf{P1.} \quad &	\max _{\boldsymbol{\Pi}}: \sum_{z=1}^{Z}\ \sum_{m=1}^{M} S^m_z( \theta_{m}G(V_{m})-R_{m})
\end{align}
such that
\begin{align}
	C1: &\ln(\theta_{m}R_{m}) -  C_{m} \ge 0, \forall m \in \{1, 2, ..., M\}, \\
	C2: &\ln(\theta_{m}R_{m}) - C_{m}  \ge \theta_{m}R_{m'} -  (C_{m}^{m'}), \nonumber\\
	& \forall m'\neq m, \ m,m' \in \{1, \dots, M\}, \\
	C3: & \sum_{z=1}^{Z}\sum_{m=1}^{M} R_{m} \leq B, \\
	C4: & T_m^{total} < T^{max},   \quad  \forall m \in \{1, 2, ..., M\},\\ 
	C5: & R_m > 0, 
\end{align}
where $C_{m}^{m'}$ is the cost of the supplier in type $m$, selecting a contract of type $m'$, and $B$ is the allocated total budget for upgrading the security defense model. In \eqref{Problem_1}, C1 is the individual rationality (IR) constraint in which each supplier has to gain non-negative utility. Constraint C2 is incentive compatibility (IC), in which each supplier should choose the exact incentive aligned with its type. In C3, the retailer guarantees that the rewards delivered to the participants do not surpass the allocated budget. The delay constraint given by C4 ensures that participants complete the assigned tasks in a defined period. In C5, each participating SPS should receive non-negative rewards. It is worth noting that the TPR can pay for the model upgrade from the security market if the required rewards are high and exceed its allocated budget. 

Finding the direct solution for \eqref{Problem_1} is extremely intricate due to the need for previous knowledge of all participants over all rounds, which is nearly unattainable. Additionally, the contracts must be adjusted based on the SPSs' participation over the training period since contributions decrease over time, and the model may converge to various stationary points, resulting in slower convergence. To address these challenges, we propose to transform the problem into an online problem in which the retailer can redesign the contracts every round. The problem is reformulated as:
\begin{align}\label{Problem_2}
\textbf{P2.} \quad &	\max _{\boldsymbol{\Pi}}: \sum_{m=1}^{M} S^m_z( \theta_{m}G(V_{m})-R_{m})
\end{align}
such that
\begin{align}
	C1: &\ln(\theta_{m}R_{m}) -  C_{m} \ge 0, \forall  m \in \{1, \dots, M\}, \\
	C2: &\ln(\theta_{m}R_{m}) - C_{m}  \ge \ln(\theta_{m}R_{m'}) -  C_{m}^{m'}, \nonumber\\
	&\forall m'\neq m, \ m,m' \in \{1, \dots, M\}, \\
	C3: & \sum_{m=1}^{M} R_{m} \leq B_z, \\
	C4: & T_m^{total} < T^{max} \quad \forall m \in \{1, 2, ..., M\},\\ 
	C5: & R_m > 0, 
\end{align}
where $B_z$ is the budget allocated for round $z$. For example, $B_z = \frac{B}{Z}$, if we aim to allocate a fixed budget in each round, it can be adjusted depending on the TPR's gain.  

Clearly, \eqref{Problem_2} is intractable, and the constraints are coupled. In particular, constraint C1 in \eqref{Problem_2} related to IR and $M(M-1)$ constraints related to IC make finding a direct solution for \eqref{Problem_2} very challenging. Thus, we start by reducing the IR and IC constraints. Then, we propose FedPot, a framework that includes tractable solutions to these challenges. FedPot is divided into two blocks, one on the SPS side and one on the TPR side. On the TPR side, we initially use the local data quality to ensure maximum utility for the TPR. Then we relax the problem to enable the TPR to select the optimal participating SPS. Next, we propose two averaging schemes considering the security aspects. To ensure fair incentives, we also proposed two reward schemes based on the claimed data quality and the contribution to the global model. 

\section{Proposed Solution}
\label{sec:proposedsol}
This section introduces the proposed solutions, including how we reduced the constraints' complexity, the prior quality determination of the local VDD, the trusted and untrusted model averaging schemes, and the rewards distribution. 
The proposed solution for the SPS side starts by deploying a honeypot on the SPS side to record all network traffic into log files. The logs are then transformed into a readable format (i.e., CSV) and cleaned with an extensive preprocessing step. {It is worth mentioning that the practice of anonymizing data (i.e., CSV files) is a commonly used method to protect sensitive information. However, several factors motivate us to use FL rather than simply anonymizing the data. First, FL mitigates the risk of data re-identification, a vulnerability in anonymization. Second, unlike anonymized CSV files, FL only shares aggregated model parameters, preserving data patterns while reducing exposure risks. Third, FL minimizes the chance of large-scale breaches by avoiding centralized data storage. Lastly, FL enables collaborative learning across distributed honeypots {without centralization, reducing communication and storage costs.}} We propose a local evaluation scheme in which each SPS evaluates its local VDD to determine whether to participate in the FL learning process. It is important to mention that the SPS is willing to participate only if the required quality is achieved. \Cref{fig:system_ill} illustrates the whole procedure performed at both the SPS and TPR.
\begin{figure*}
    \centering
    \includegraphics[scale=0.5]{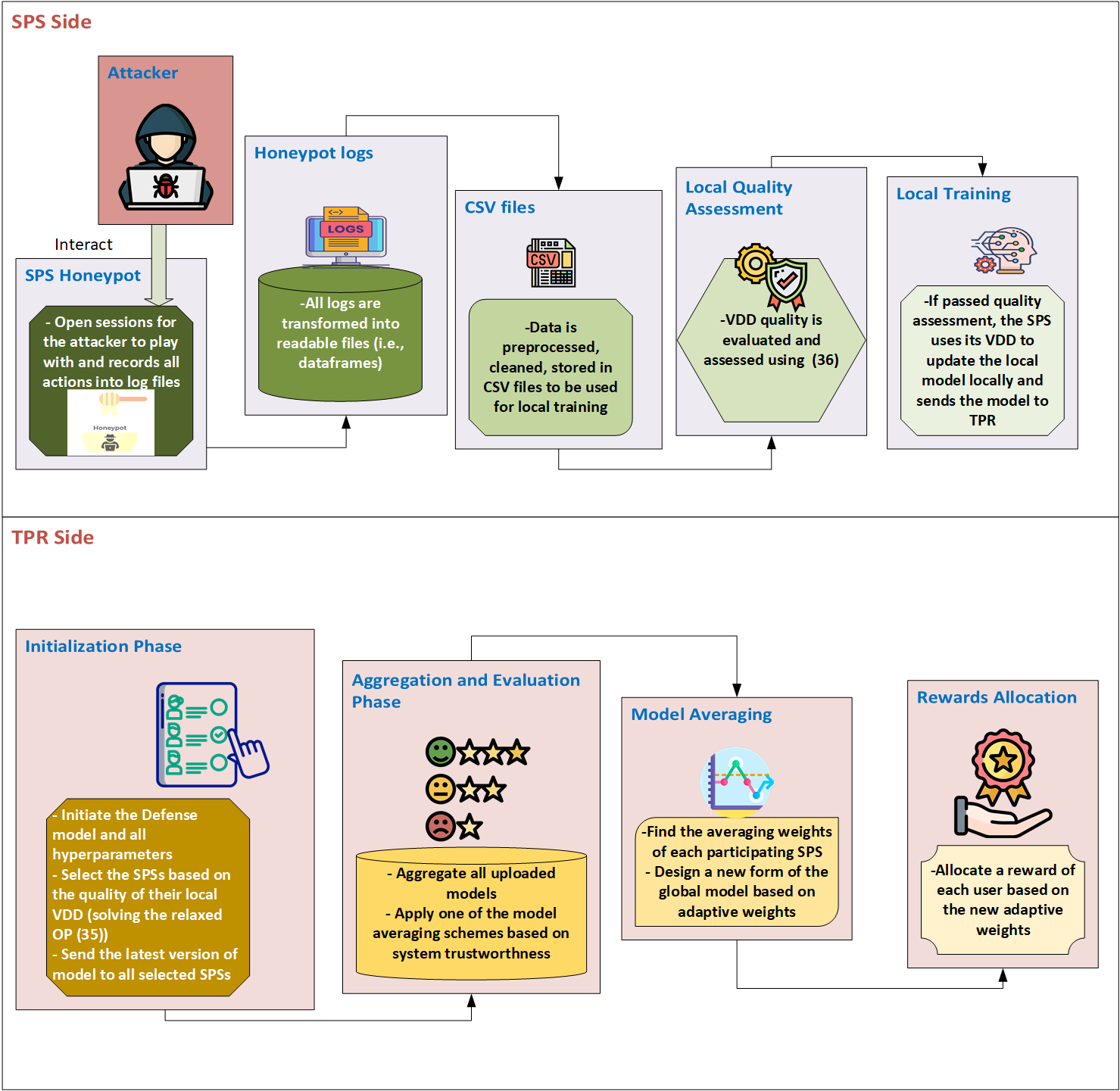}
    \caption{FedPot Architecture.}
    \label{fig:system_ill}
\end{figure*}

\subsection{Constraints complexity reduction}
As stated previously, it is almost intractable to directly solve \eqref{Problem_2}. To simplify, we present the subsequent lemmas.
\begin{lemma}
\label{lemma1}
Given the budget for each round $B_z$, for any feasible contract $(\boldsymbol{V}_m, R_m)$, $R_m \geq R_{m'}$ $\iff$ $m \geq m'$ $\forall \quad m, m' \in \{1, 2, ..., M\}$
\end{lemma}
\begin{proof} 
we first prove that if $\theta_m \geq \theta_{m'}$ where $m \geq m',$ then $R_m \geq R_{m'}$. Adding the IC constraints for both types $\theta_m$ and $\theta_{m'}$, yields:
\begin{equation}
\label{pro1}
    \ln(\theta_m R_m ) - C_m \geq \ln(\theta_m R_{m'}) - C_{m'},
\end{equation}
and 
\begin{equation}
\label{pro2}
    \ln(\theta_{m'} R_{m'}) - C_{m'} \geq \ln(\theta_{m'} R_{m}) - C_{m},
\end{equation}
We add both \eqref{pro1} and \eqref{pro2}, we have:
\begin{equation}
\label{pro3}
    \ln(\theta_m - \theta_{m'})(R_m-R_{m'}) \geq 0.
\end{equation}
Thus, $R_m \geq R_{m'}$.
\end{proof}
From \Cref{lemma1}, we note that more rewards will be given to more participants. This means that \Cref{lemma1} is monotonic. Hence, this analysis of IC constraints  {reduces} the IR constraints as indicated in the following lemma.
\begin{lemma}
\label{lemma_2}
Given the $B_z$, and \Cref{lemma1} related to IC constraints and the sorted participants based on their data quality, the IR condition can be reduced as:
\begin{equation}
    \ln(\theta_{1}R_{1}) -  C_{1} \ge 0 
\end{equation}
\end{lemma}
\begin{proof}
\label{pro4}
Given the sorted SPSs based on their types as defined in \Cref{sec:system_model}, the IC constraints are used as follows:
\begin{equation}
    \ln(\theta_m R_m ) - C_m \geq \ln(\theta_m R_1) - C_1 \geq \ln(\theta_1 R_1) - C_1 \geq 0.
\end{equation}
From \eqref{pro4}, we note that if the first SPS's type meets the IR constraint, all other SPSs' types will automatically meet other IR constraints. Therefore, the IR of type $1$ is sufficient to achieve all other IR constraints.
\end{proof}
According to \Cref{lemma1,lemma_2}, we can transform the IC constraint into  local downward incentive compatibility (LDIC) as follows~\cite{icentive_FL1}:
\begin{equation}
\label{eq:LDIC}
    \ln(\theta_{m}R_{m}) - C_{m}  \ge \ln(\theta_{m}R_{m-1}) -  C_{m-1}, \forall n \in \{2, 3, ..., M\},
\end{equation}
and local upward incentive compatibility (LUIC) as follows:
\begin{equation}
\label{eq:LUIC}
    \ln(\theta_{m}R_{m}) - C_{m}  \le \ln(\theta_{m}R_{m+1})  -  C_{m+1}, \forall n \in \{1, 2, ..., M-1\}.
\end{equation}
From \Cref{eq:LDIC,eq:LUIC}, we note that participants should be given rewards based solely on their contributions. Moreover, we observe that the objective function in \eqref{Problem_2} is decreasing w.r.t $R_m$ and increasing w.r.t $G(V_{m})$. 

\subsection{Quality Measure of the SPSs' VDD}
\label{Sub:data_quality}
The SPSs asymmetrically exchange information, sharing only the local model updated using the collected VDD with the TPR, rather than the VDD itself. Hence, the TPR is necessary to validate the uploaded models.
The quality of the local model depends on several factors, including the variety of the data, the number of contained classes, {and the number of updates performed}.
In this regard, we first model the SPS's data by considering the data quality, {and }, then by assessing the impact of the local models on the generalization of the global model. This strategy allows the TPR to guarantee a robust model even in the presence of malicious participants.
Each local dataset $\mathcal{D}m$ generally comprises traffic data with input-output pairs of $\{\mathbf{x}_{i,d}^{(m)},y_i^{(m)}\}^{D_m}_{i=1}$, where $\mathbf{x}_{i,d}^{(m)} \in \mathbb{R}^d$ denotes the input holding $d$ attributes, and $y_i^{(m)} \in \mathbb{R}$ denotes the matching class label. 
Given that every SPS has its own network traffic patterns and a varying response according to its activity, the honeypot logs produce data of various sizes and follow a non-identical and independently distributed (non-i.i.d.) manner.
Furthermore, the logs in each SPS may have distinct attack types, and some of the SPSs might not have any malicious samples.
Let $\mathcal{D} = \{\mathcal{D}_1 \cup \mathcal{D}_2 \cup \ldots, \cup~\mathcal{D}_M\}$ represent the aggregated dataset from all SPSs. We define the probability of a sample, $i$, being included in the local logs as:
\begin{equation}
\rho(\mathcal{D}_m|i) = \left\{
\begin{array}{l}
1: \mathrm{i \in  \mathcal{D}_m}\\
0: \mathrm{otherwise}
\end{array}.
\right.
\label{eq1}
\end{equation}
\begin{definition}[$\delta$-Data Coverage] With a predetermined radius, the coverage of data collection through sampling is obtained by: 

	\begin{equation}
	\rho(\mathcal{D}_m,\delta)= \mathcal{D} \cap \cup_{x_{i} \in \mathcal{D}_m} \Lambda(x_{i}, \delta),
\end{equation}
where $\Lambda(x_{i}, \delta)$ is an open ball space with a radius of $\delta$ centered at $x_i$.
\end{definition} 
\noindent Considering the data space among all SPSs is {an Euclidean} space, the range of $\delta$ is confined to the interval $[0,\sqrt{d}]$.

\begin{definition}[VDD Quality], $\varphi(\mathcal{D}_m)$ gives the estimation of the local VDD quality measurements defined as:
\begin{equation}
	\varphi(\mathcal{D}_m)=\frac{1}{\sqrt{d}}\int_{0}^{\sqrt{d}} \rho(\mathcal{D}_m,\delta) \ {\rm d}\delta 
\end{equation}
\end{definition} 

\begin{remark}
The value of $\varphi(\mathcal{D}_m)$ can indicate the quality of the local VDD. A more increased value of $\varphi(\mathcal{D}_m)$ implies improved data diversity, resulting in a better quality of the locally updated model. This is because greater spaces between traffic samples grasp more patterns during model training, enhancing the overall generalization of the global model. Let $\phi={\varphi_1,...\varphi_{M}}$ be the quality set of all SPS with $M$ types, and all SPSs with $\varphi(\mathcal{D}_m)\in[\frac{m-1}{M},\frac{m}{M}]$ hold a data of type $m$. \end{remark}
To solve \Cref{Problem_2}, the TPR can leverage each participant's previous performance, which might be time-consuming since assessing all upgraded models is a post-processing action. As a result, the quantity of $\varphi(\mathcal{D}_m)$ can be utilized to measure local VDD quality before a given SPS is selected.
This approach allows the TPR to select high-quality data and avoid using models from participants with low-quality data, which can negatively impact the global model's performance. By employing this technique, the TPR can ensure that the final model is robust, even in the presence of malicious participants.

We can substitute  $\pi_m=(\boldsymbol{V}_m, R_m)$ in \Cref{sec:system_model} by $\pi_m=(\boldsymbol{\Omega}, R_m)$. Accordingly, we can further reformulate the optimization problem in \eqref{Problem_2} as follows:
\begin{align}\label{Problem_3}
\textbf{P3.} \quad &		\max _{\boldsymbol{\Pi}}: \sum_{m=1}^{M} S^m_z\varphi(\mathcal{D}_m) 
\end{align}
such that
\begin{align}
	C1: &\ln(\theta_{1}R_{1}) -  C_{1} \ge 0,  \\
	C2: & \ln(\theta_{m}R_{m}) - C_{m}  \ge \ln(\theta_{m}R_{m-1}) -  C_{m-1}, \forall m, \\
	C3: &   \ln(\theta_{m}R_{m}) - C_{m}  \le \ln(\theta_{m}R_{m}) -  C_{m+1}, \forall m, \\
	C4: & \sum_{m=1}^{M} R_{m} \leq B_z, \\
	C5: & T_m^{total} \leq T^{max}  \quad \forall m \in \{1, 2, ..., M\},\\
	C6: & R_m > 0.
\end{align}
{It is worth mentioning that the optimization problem (\textbf{P3}) considers the lure score to adaptively place honeypots in the network to maximize their luring potential. As in our proposed approach, the lure score can be calculated based on various factors such as network traffic patterns, historical attack vectors, and even machine learning models trained for this specific purpose. \Cref{fig:Schematic} illustrates how the luring score can be determined.  The system ensures that honeypots are effectively deployed by considering the lure score, which will most likely capture meaningful data and contribute to better FL. Specifically, high-quality data, signifying that the user’s honeypot has successfully deceived the attackers, is defined by a luring score. This leads to the development of more accurate and comprehensive models, which are essential for detecting and mitigating sophisticated cyber threats.}
\begin{figure}
    \centering
    \includegraphics[width=0.9\linewidth]{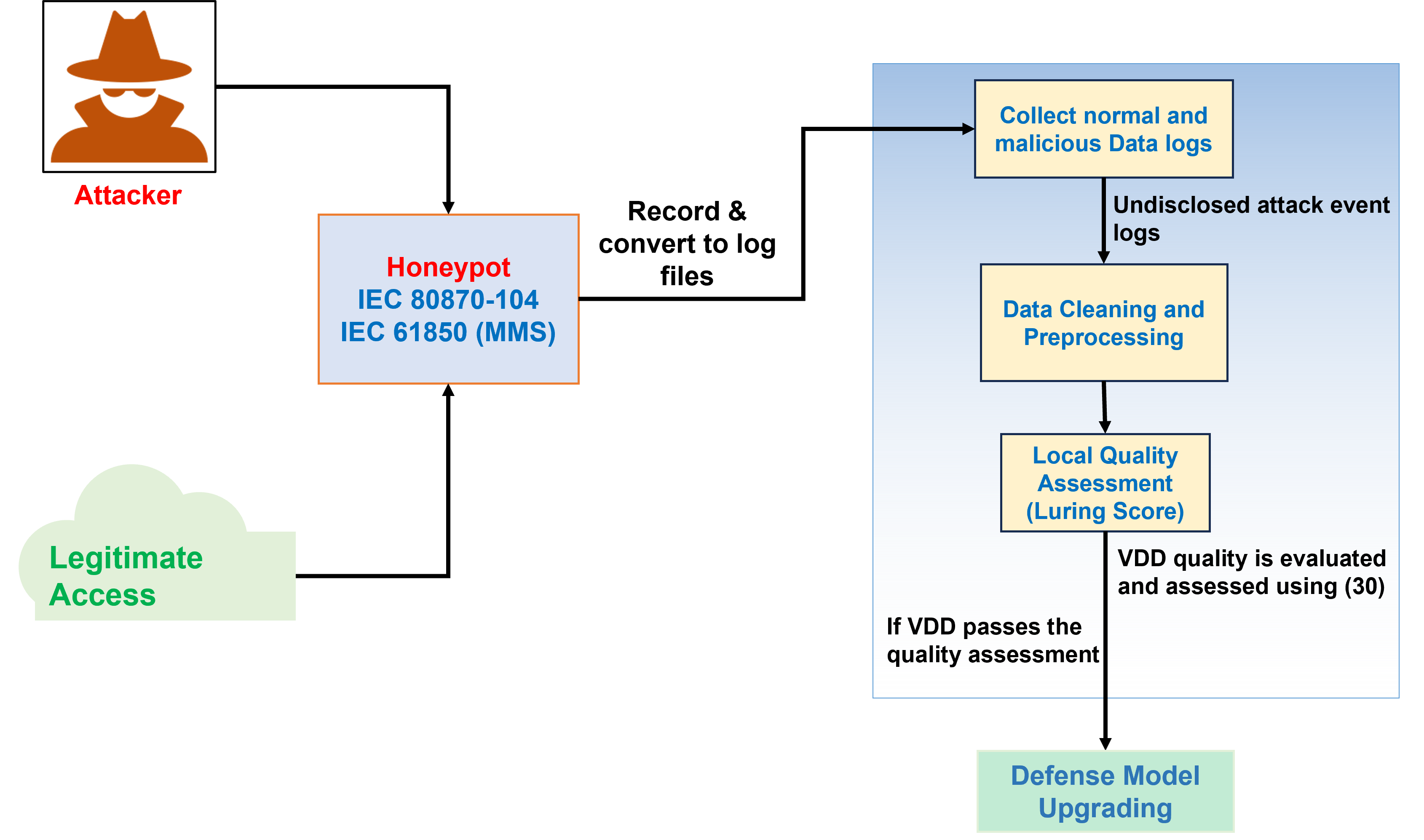}
    \caption{{Schematic illustrating how the luring score can be determined}}
    \label{fig:Schematic}
\end{figure}

However, we note that the problem in \eqref{Problem_3} is still challenging to be solved directly due to the monotonic constraints in C1, C2, and C3. The authors in \cite{icentive_FL1} proposed relaxing such constraints, enabling the problem to be directly solvable as an optimization problem. However, in the context of FL, it is not practical to apply such solutions due to the following: 
\begin{itemize}
    \item The server cannot adequately analyze local updates even if all selected SPSs are trusted. The VDD quality metric can only be utilized as described in \Cref{Sub:data_quality} only if all participants are completely trusted.
    
    \item Similar data may have varying contributions during the global training rounds, wherein it could initially make a higher contribution but gets reduced over time. For instance, the training process slows down if the model is closer to its stationary point, regardless of the associated cost or local data quality.
    Thus, the rewards should be given based on the value added to the global model every round, not as a fixed contract.
    \item The server should ensure fairness between the participants where the higher cost does not imply higher quality.
\end{itemize}

In the following section, we solve this problem iteratively by implementing a smooth FL process, including novel selection, aggregation, and averaging schemes. Then, we design a fair and efficient postprocessing rewards mechanism that fully ensures fairness so that the constraints (36-39) can be implicitly satisfied. The introduced solutions incorporate three phases; the first phase (i.e., the Preprocessing Phase) aims to solve an optimized relaxation problem to select the proper participants. The selected participants receive the global model, update it locally, and upload it to the TPR. In the second phase, the locally upgraded models by all participating SPSs are aggregated. The TPR evaluates the received models using two schemes (discussed below) and assigns a contribution rank for each model. In the third phase, the enhanced FedAvg based on adopted weights is applied under two scenarios when the participants are fully trusted and when some participants are malicious. The rewards are then given based on the contribution rank given to each participant (i.e., based on the uploaded model). 

\subsection{Quality-Assurance Model Averaging}
\label{scheme1}
The phases of the proposed solution are as follows. By resolving \eqref{Problem_3} during the initial phase (i.e., the preprocessing phase), it attempts to identify the appropriate SPSs. The identified participants will then receive the global model for local updating. The second phase combines the local updates from all SPSs to mitigate the challenges arising using the traditional FedAvg approach being assessed in two different settings. 

\paragraph{Trust-based Model Averaging}
The TPR trusts all participants, and the weight of each participant is determined according to their reported data quality:
\begin{equation}
\label{eq:trusted}
    v_m = \frac{\varphi(\mathcal{D}_m)}{\sum_{m=1}^{M}\varphi(\mathcal{D}_m)}
\end{equation}
Following that, the global model is generated and updated given by 
\begin{equation}
{\mathbf{\boldsymbol{V}(z)}={\sum_{m=1}^{M}v_m \mathbf{\boldsymbol{V}}_{m}}\label{eq:globalAveragemode1}.}
\end{equation}
In contrast to the conventional FedAvg, which utilizes the amount of data samples, this enhanced FedAvg does model averaging by weighting the local models according to the local data quality. Practically, weighting the model parameters according to the amount of data is futile since the data may contain redundancies and slightly affect the global model. Further details can be found in  \Cref{sec:eval}, where the conventional FedAvg approach demonstrates inferior performance.

\subsection{Two-steps verification mechanism}
\label{scheme2}
In this approach, the retailer initially selects the participating suppliers based on their local data quality. As explained in \Cref{Sub:data_quality}, the higher the data quality, the better the model update. Thus, the server adopts the value of $\varphi(\mathcal{D}_m)$ to prioritize the suppliers and select the best set. However, the malicious supplier may claim higher data quality. Therefore, the server needs to apply a further verification step by testing the model update before confirming the supplier's reward. This can be done by estimating the Euclidean distance between two consecutive updates of the global model and the updated local model, respectively, or by having global data samples and feeding them into the uploaded models one by one before conducting the global averaging. 

\paragraph{Untrust-based Model Averaging} In this scenario, some participating SPSs might declare to have valuable data, yet the shared models may be perturbed or generated randomly. Consequently, a two-step verification is proposed, where the TPR initially assesses the obtained model employing the Euclidean distance or utilizing generalized test data evaluating every model independently. Assuming the latter is utilized, the adjusted averaging weights can be determined using the formula as follows: 
\begin{equation}
\label{eq:untrusted}
    w_m = \frac{G(v_{m})}{\sum_{m=1}^{M}G(v_{m})}
\end{equation}
Following that, the model is generated and updated based on
\eqref{eq:globalAveragemode1}.

\subsection{Incentive Reward design}
\label{scheme3}
We propose applying the Soft-max function to determine the incentives provided to the selected SPSs after computing the authentic average weights employing one of the methods outlined in the preceding section. This will guarantee equality according to the contribution provided by each SPS. It can be expressed as follows:
\begin{equation}
    \label{eq:rewardcal}
    w_m = \frac{e^{v_m}}{\sum_{m=1}^{M}
    e^{v_m}}. 
\end{equation}
It is important to note that for each round $z$, the total rewards will be less than the budget assigned. As a result, the participants receive the earned rewards from:
\begin{equation}
\label{eq:rewads_given}
    R_m = w_m B_z.
\end{equation}
We may infer from \eqref{eq:rewardcal} and \eqref{eq:rewads_given} that the obtained rewards ensure equality for all SPSs participants and completely meet the assigned budget. {It is worth noting that our approach considers the heterogeneity of threats and their distribution implying that most SPSs, over a reasonable time frame, are exposed to some kind of malicious activity. It might be accurate that at certain times, specific SPS face a higher volume or intensity of attacks, but when viewed over an extended period, the distribution becomes relatively even. We have designed the reward mechanism to factor in not just the volume but also the quality of the data. This way, even if an SPS does not face a high number of attacks, the uniqueness or novelty of a single quality data point they contribute could be of immense value, ensuring they are adequately rewarded.} Alg.  \ref{Alg1} outlines the overall process of the suggested approach.
In Alg. \ref{Alg1}, the TPR starts by initializing the model and the hyper-parameters (step 1), then gathering prior information from all SPSs (step 2). From steps 3–17, the TPR runs the FL algorithm by determining the deadline for each round (step 4) and solving the relaxed \textbf{P3} (step 5) to select the best participants while accounting for the stated data quality. The global model is then sent (step 6), and all selected participants undertake local training and submit all local models once finished (steps 7-10). The TPR aggregates all models (step 11) before deciding on one of the proposed schemes. If the TPR can confirm the trustworthiness of the participating SPSs, the solution in \Cref{scheme1} is applied (steps 12 and 13), and the solution in \Cref{scheme2} is applied (steps 15 and 16) if TPR can not. Finally, as explained in \Cref{scheme3}, the rewards are fairly given to the participants efficiently. 

\begin{algorithm}[t]
\footnotesize
    \caption{FedPot Framework}  \label{Alg1}
    \LinesNumbered
    \KwIn{All available SPSs $M$}
    \KwOut{Upgraded Security Model $V$}
    \textbf{Initialize:}
     starting global model $V_0$, local epochs $\varepsilon$, step size $\eta$, rounds $Z$, and assigned budget for every global round\;
  TPR collects previous information (e.g, $\varphi(\mathcal{D}_m)$, from existing SPSs $M$.\;
 \For{$z=1$ to $Z$}{
 The TPR establishes the time limit (i.e., deadline)\;
 TPR addresses the optimization problem, \textbf{P3}, to choose the best entities for participating in the model update\;
 TPR shares the model $V_{z-1}$ to chosen SPSs\;
 \For{Each SPS $m \in M$ synchronously}{
 SPS $m$ gets $V_{z-1}$\;
  SPS $m$ updates $m$ employing its VDD obtained data for $E$ epochs\;
  SPS $m$ uploads $V_m$ to TPR\;
 }
  The TPR gathers all submitted updates from participating SPSs\;
  \If{Every participant is trusted}
  {
  TPR employs Equation \eqref{eq:trusted} to recompute the weight associated with each update.
  }
  \Else 
  {
  TPR assesses each model using the generalization test data\;
  TPR utilizes Equation \eqref{eq:untrusted} to recompute the weight for every update
  }
  TPR utilizes equation (39) to generate a modified global model.
 }
\end{algorithm}
\subsection{{Architectural Mapping of the Proposed Solution}}
{It is worth noting that our control framework is designed to emulate the hierarchical architecture commonly observed in modern SG systems. It comprises three main layers: the Data Acquisition Layer, the Control Layer, and the Decision Layer. 
\textbf{Data Acquisition Layer}: This is the base layer where all sensors and IoT devices (as represented by our use of the N-BaIoT dataset in Section VI) are located. These devices are responsible for collecting real-time information such as voltage, current, and frequency from different points in the grid.
{\textbf{Control Layer}: where the real-time data is analyzed and control signals are generated. This layer integrates the security protocols as studied through the IEC 104 dataset in section VI-C. The layer includes controllers like Remote Terminal Units (RTUs) and programmable logic controllers (PLCs), which interact directly with the devices in the Data Acquisition Layer. The control layer continuously monitors and analyzes real-time data, identifying potential cyber threats or operational anomalies. This layer generates signals, executes responses to mitigate threats, and integrates various security protocols. A key feature is the strategic deployment of honeypots within the network. These honeypots serve as decoy systems designed to attract cyber attackers. This allows for an in-depth analysis of attack strategies, significantly enhancing the grid's cyber resilience and providing valuable insights into potential vulnerabilities.}
\textbf{Decision Layer}: This is the topmost layer, consisting of control centers or cloud-based systems where higher-level decision-making processes occur. Here, our proposed anomaly detection algorithms and security measures are implemented (i.e., ML-based detector).  
To generate realistic grid scenarios, we employed both the N-BaIoT and IEC 104 datasets, creating a diverse set of operating conditions and cyber-attack vectors. Our framework is also designed to be scalable and robust, which allows for the integration of additional sensors and control units as required. The proposed incentive-based model is implemented at the Decision Layer, ensuring that it benefits from the real-time data collected and analyzed at the lower layers and addressing the free-rider problem. This facilitates more effective and timely decision-making. It is worth noting that in this layer of the context of cybersecurity, test datasets cannot remain static. When a new type of attack is detected by a honeypot, it may initially be evaluated through heuristics, expert rules, and anomaly detection methods for interim validation. Contextual information provided by the SPS also helps the TPR to assess the contribution's potential impact. To corroborate new threats, we propose sharing them (anonymously) with a subset of trusted SPS for additional verification. As these new types of attacks are validated, they will be incorporated into future iterations of the test data, ensuring a relevant benchmark for subsequent evaluations.}

\subsection{{Theoretical Analysis of Robustness Against Malicious SPSs}}
\label{sec:robustness_theory}
{To validate the resilience of our Quality-Assurance Model Averaging and Two-step Verification Mechanism against malicious activity, we introduce some theoretical metrics and analysis.
Let \( \Delta V(z) \) denote the deviation of the global model \(\boldsymbol{V}(z)\) under the Trust-based Model Averaging scheme given by ~\eqref{eq:globalAveragemode1}:
\begin{equation}
    \Delta V(z) = \left\| \boldsymbol{V}(z) - \boldsymbol{V}^* \right\|
\end{equation}
Here, \(\boldsymbol{V}^*\) represents the optimal global model that would have been achieved without any malicious activity. We aim to bound \(\Delta V(z)\):
\begin{equation}
    \Delta V(z) \leq f\left(\varphi(\mathcal{D}_m), M, v_m, \ldots \right)
\end{equation}
Here, \(f\) is a function that encapsulates the contributions of local data quality \(\varphi(\mathcal{D}_m)\), the total number of participants \(M\), and their respective weights \(v_m\), among other parameters. Similarly, we define \(\Delta W(z)\) as the deviation under the Untrust-based Model Averaging scheme given by ~\eqref{eq:untrusted}.
\begin{equation}
    \Delta W(z) = \left\| \boldsymbol{V}(z) - \boldsymbol{V}^* \right\|
\end{equation}
We aim to bound \(\Delta W(z)\) as follows:
\begin{equation}
    \Delta W(z) \leq g\left( G(v_m), M, w_m, \ldots \right)
\end{equation}
Here, \(g\) is a function capturing the effects of the generalized test \(G(v_m)\), total number of participants \(M\), and the adjusted weights \(w_m\), among others. In essence, the bounded nature of \(\Delta V(z)\) and \(\Delta W(z)\) implies that our proposed schemes are robust against a variable number of malicious SPSs, maintaining the global model's efficacy.}

\section{Performance Evaluation}
\label{sec:eval}
In this section, we establish an experimental setting to evaluate the efficacy of our proposed methods for improving the security defense model and allocating fair incentives.

\subsection{Experimental Setup}
\label{sec:experiments}
\begin{figure*}
  \centering
  \begin{subfigure}{.32\linewidth}
    \includegraphics[width=\linewidth]{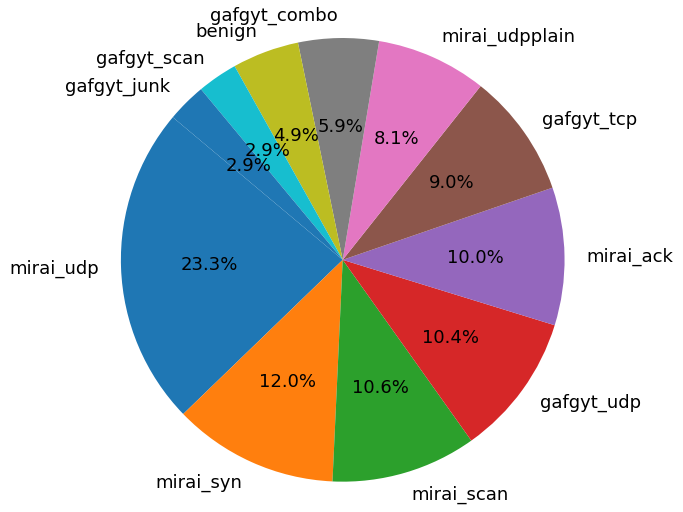}
    \caption{Device-1} 
    \label{fig:sub1}
  \end{subfigure}\hfill
  \begin{subfigure}{.32\linewidth}
    \includegraphics[width=\linewidth]{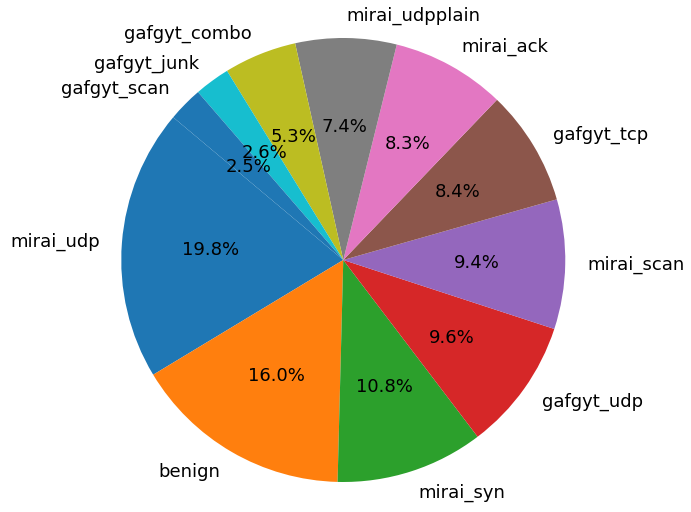}
    \caption{Device-2} 
    \label{fig:sub2}
  \end{subfigure}\hfill
  \begin{subfigure}{.32\linewidth}
    \includegraphics[width=\linewidth]{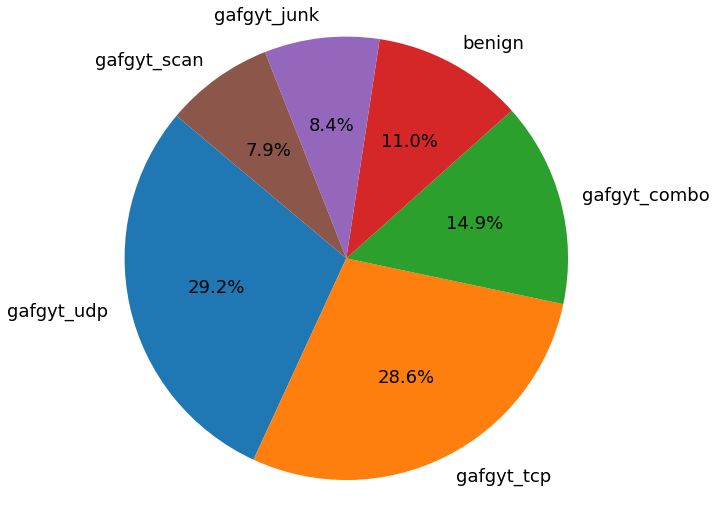}
    \caption{Device-3} 
    \label{fig:sub3}
  \end{subfigure}
  \caption{{Label distribution among the logs of the first three devices.}}
  \label{fig:label_dist}
\end{figure*}
{We employ the N-BaIoT \cite{meidan2018n} as well as IEC 104 and IEC MMS datasets\cite{dataset}. N-BaIoT was initially designed to cover the security aspects of IoT devices. This dataset is crucial for modern SG systems increasingly relying on IoT sensors and actuators for efficient and intelligent grid management.  It includes data from benign operations as well as from a range of attacks like Mirai and Gafgyt (Bashlite). Mirai attacks focus on different flooding techniques using ACK, SYN, UDP, and other protocols while performing device vulnerability scans. Gafgyt attacks, similarly, engage in device scanning and deploy a mix of attack strategies, including TCP and UDP flooding. As in \Cref{fig:label_dist} for the first three devices, it consists of eleven classes, one benign, while the remaining ten represent adversarial classes (i.e., indicating malicious data), with each sample including 115 attributes. These datasets, in particular, include DDoS, which causes networks to become overloaded with traffic; bot attacks that take advantage of device flaws; man-in-the-middle attacks that intercept communications; SQL injection that manipulates databases; and XSS which injects malicious scripts. Firmware flashing, physical tampering, and side-channel attacks also represent advanced threats. In SG, DDoS attacks can disrupt grid communications, Bot attacks may compromise grid control systems, and MitM attacks pose risks to data integrity. SQL injection and XSS are relevant where SGs use web interfaces or databases. Firmware flashing and physical Tampering highlight the need to secure grid hardware, while side-channel attacks demonstrate the importance of protecting against indirect data leaks.}
%

We compare our proposed method with the traditional FedAvg technique, where the server allocates weights to the local model updates based on the amount of traffic data samples they include. Furthermore, we evaluate the incentives-based algorithm against a customized-based algorithm.
%
We employ a supervised deep neural network (DNN) model consisting of 115 input layers followed by 115, 62, and 32 hidden layers in which $\frac{1}{2}$, $\frac{1}{4}$, and $\frac{1}{8}$ inputs follow each hidden layer, respectively. The cross-entropy is used as a loss function. 
%
We initially consider the N-BaIoT dataset's initial device count, consisting of nine devices exhibiting non-i.i.d. characteristics. 
Similar parameters are adopted for all experiments. There are eight SPSs as participants in the FL configurations, with each SPS possessing one out of the nine devices present in the N-BaIoT datasets, including both benign and adversarial samples. One device is solely utilized to evaluate the overall model developed by each of the other eight participants. We conduct each experiment five times to ensure the accuracy of the findings. Then, 50 devices are added to the data distribution. Every device uses a $10$ local epoch and a $32$ local batch size. We employ an adjustable learning rate with a default value of $0.01$.

 We assess the efficiency of our proposed scheme using the following metrics accuracy, loss, true positive rate (TPRate), false positive rate (FPR), and F1-measure given by: 
$
    TPRate =\frac{TP}{TP+FN}, 
$ $
    TNR = \frac{TN}{TN+FP}
$, and $
   F1{-}score = \frac{2 * Precision * Recall}{Precision + Recall}
$. To assess the model's efficacy, we utilize the testing accuracy, which measures the proportion of accurately identified used incursions as a performance indicator.  
Finally, we use rewards distribution fairness to showcase the efficiency of the proposed rewards allocation framework.

\subsection{Numerical Results on BaIoT dataset}
In this section, we perform our numerical analysis on IDD and non-i.i.d assumptions.

\subsubsection{The IID Setting}

\Cref{fig:accu_comparetrusted_iid} provides an overview of the trust-based model averaging framework performance under several distinct settings: complete data offloading to TPR (centralized training), the conventional FedAvg, and the proposed enhanced-FedAvg (FedAvg based on data quality).
The results of the proposed algorithm demonstrate that it outperforms the conventional FL approach, which can be attributed to the effect of assigning model weights according to the local data quality instead of the data size as in conventional FedAvg. Furthermore, the enhanced FedAvg scheme can attain almost traditional centralized performance. In contrast, the traditional FedAvg algorithm needs more rounds to obtain the required {accuracy,} which increases the cost, leading to surpassing the budget assigned.

\begin{figure}[t]
    \centering
        \includegraphics[width=0.45\textwidth]{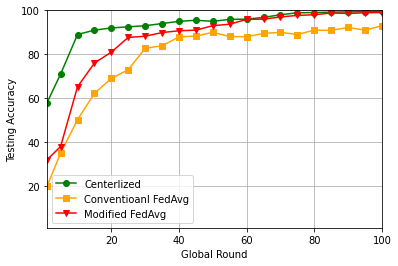}
         \caption{Comparative Analysis of Test Accuracy: Enhanced FedAvg, Conventional FedAvg, and Centralized-Based Algorithms.}
            \label{fig:accu_comparetrusted_iid}
  \end{figure}
  
  \begin{figure*}[t]
    \centering
       \begin{subfigure}[b]{0.32\textwidth}
    \includegraphics[width=\textwidth]{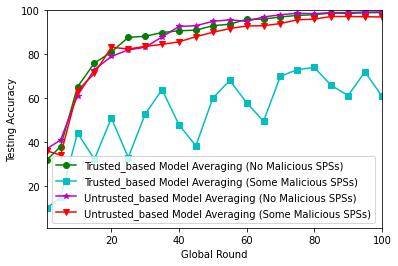}
    \caption{9 SPSs}
    \label{fig:accu_compareuntrusted1_iid}
     \end{subfigure}
    \begin{subfigure}[b]{0.32\textwidth}      
    \includegraphics[width=\textwidth]{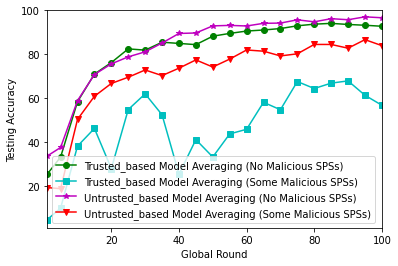}
    \caption{50 SPSs}
    \label{fig:accu_compareuntrusted2_iid}
     \end{subfigure}
    \begin{subfigure}[b]{0.32\textwidth}
    \includegraphics[width=\textwidth]{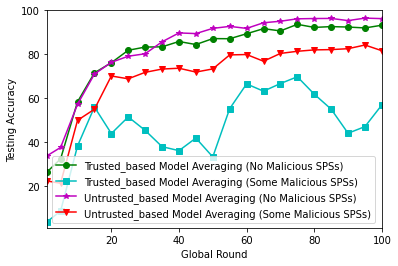}
    \caption{100 SPSs}
    \label{fig:accu_compareuntrusted3_iid}
     \end{subfigure}
         \caption{Test Accuracy Comparison of Proposed Schemes in the Presence of Malicious SPSs Claiming High-Quality Data While Uploading Defective Models (i.i.d. Data Distribution, N-BaIoT dataset).}
  \end{figure*}
\Cref{fig:accu_compareuntrusted1_iid} shows the impact of the untrust-based model averaging framework under two scenarios: when all participants are completely trusted and when some are malicious. Our observations reveal that depending solely on the data quality asserted by SPSs, the presence of corrupted models leads to subpar models. This outcome arises due to malicious participants who assert high-quality data yet disperse it arbitrarily or manipulate local models.
The suggested untrusted model averaging approach may efficiently remove detrimental models by assigning weights to the shared models according to the generalization test data.
This causes the effects of adversarial or disturbed models to vanish during the averaging phase or at least reduces their negative impact on the designed global model utilized for NIDS. 
Contrarily, malicious participants may significantly reduce the defensive model's effectiveness, despite the data quality they claimed throughout the selection process. We repeat similar experiments while splitting the data between 50 and 100 SPSs to showcase the effectiveness of the proposed approaches. 

As shown in \Cref{fig:accu_compareuntrusted2_iid,fig:accu_compareuntrusted3_iid}, we use similar settings but distribute the data across 50 and 100 SPSs and keep the labels in an i.i.d fashion. We follow the original data distribution, where the labels across the SPSs follow the same distribution. We observe that even if the data is i.i.d, the presence of some malicious participants affects the performance of the models, especially as the number of SPSs increases. Nevertheless, the untrust-based averaging scheme still performs better than the trust-based scheme. Yet, the accuracy is slightly decreased as the number of SPSs increases. 

Furthermore, we assess the efficacy of the proposed approaches regarding the fairness of the reward using the true positive rate (TPRate) and the true negative rate (TNR), critical performance indicators for the rewards allocation and detection rate, and in the security defense model. \Cref{tab:TPR-TNR_iid} shows that, whether malicious participants are present or not, both proposed approaches exceed the traditional FedAvg. Nonetheless, in the presence of malicious participants, the proposed trust-based model averaging (PTrusted) performance exhibits a considerable decline in TPRate, TNR, and fairness of rewards. This degradation occurs as we solely depend on the claimed data quality. 
On the other hand, the proposed untrust-based model averaging ensures adequate outcomes despite malicious participants. This is achieved through a two-step verification process that excludes altered models and accordingly assigns rewards, thereby achieving a high level of fairness. Meanwhile, FedAvg distributes the incentives randomly depending on the data size.
 \begin{table}[t]
     \centering
      \caption{Comparative Results for TPRate and TNR of Conventional FedAvg and Proposed Solutions Under Attack and Non-Attack Scenarios (i.i.d. Data Distribution, N-BaIoT dataset).}
     \begin{tabular}{|p{1.3cm}|p{0.5cm}|p{0.4cm}|p{0.5cm}|p{0.5cm}|p{0.7cm}|p{0.5cm}|}
     \hline
         Alg. & \multicolumn{2}{p{0.1cm}} Con. FedAvg & \multicolumn{2}{|p{0.1cm}} P-Trusted &   \multicolumn{2}{|p{0.1cm}} P-Untrusted \\ \hline
         Malicious SPSs & No & Yes &  No & Yes &  No & Yes \\ \hline
        TPRate & 89.7 & 52.9 & 92.95 & 76.3 & 95 & 93\\\hline
        TNR & 62.1 & 30.02 & 85 & 61.9 & 90.85 & 88\\ \hline
        Fairness & 50.01 & 31 & 93.9 & 55 & 95.1 & 94.7\\ \hline
     \end{tabular}
     \label{tab:TPR-TNR_iid}
 \end{table} 
 
\begin{figure*}[htbp]
    \centering
       \begin{subfigure}[b]{0.30\textwidth}
    \includegraphics[width=\textwidth]{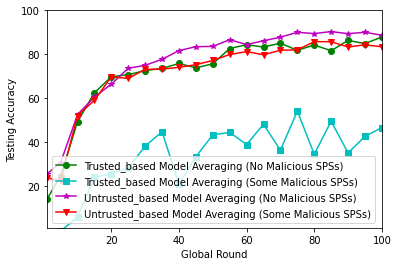}
    \caption{}
    \label{fig:accu_compareuntrusted1}
     \end{subfigure}
            \begin{subfigure}[b]{0.30\textwidth}
    \includegraphics[width=\textwidth]{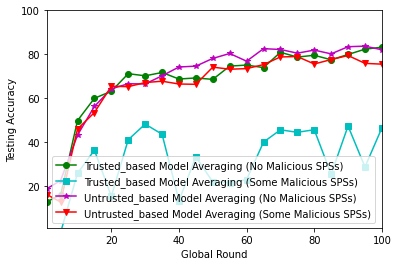}
    \caption{}
    \label{fig:accu_compareuntrusted2}
     \end{subfigure}
    \begin{subfigure}[b]{0.30\textwidth}
    \includegraphics[width=\textwidth]{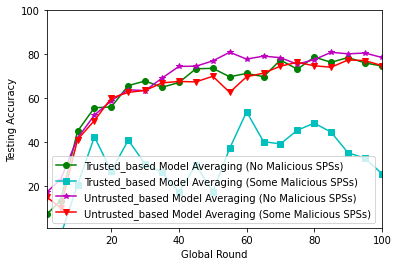}
    \caption{}
    \label{fig:accu_compareuntrusted3}
     \end{subfigure}
         \caption{{Test Accuracy Comparison of Proposed Schemes in the Presence of Malicious SPSs Claiming High-Quality Data While Uploading Defective Models (non-I.I.D. Data Distribution, N-BaIoT dataset).}}
  \end{figure*}

\subsubsection{The Non-IID setting}
For the non-i.i.d. setting, we assume that each SPS holds only a maximum of 2 types of attacks (i.e., Mirai UDP and combo). We consider three scenarios regarding the number of devices: 9, 50, and 100. We assign only a maximum of two classes for each SPS to ensure that the datasets amongst SPSs are non-i.i.d. The default data distribution is considered non-i.i.d. However, as illustrated in \Cref{fig:label_dist}, the labels' distribution across all devices is almost identical, even for those holding only 6 classes. Hence, in this work, we consider a realistic non-i.i.d. distribution to evaluate the performance in the worst situations. 

As shown in \Cref{fig:accu_compareuntrusted1}, we initially run the first scenario when the number of devices is only 9. One can see that the performance in general drops by approximately 15\% for both proposed averaging schemes. However, when some participants maliciously claimed high-quality local logs during the selection phase and uploaded perturbed models, the performance of the first proposed scheme (i.e., trust-based model averaging) dramatically dropped by almost 50\% to reach 46\% accuracy. This is due to adopting inaccurate weights during the evaluation phase, thereby impacting the performance of the global model. In contrast, the second proposed scheme (i.e., untrust-based model averaging) shows immunity against the perturbed models even when the data is non-i.i.d. 

In \Cref{fig:accu_compareuntrusted2,fig:accu_compareuntrusted3}, we repeat the same experiments, distributing the data across 50 and 100 SPSs, respectively. From a security perspective, the performance in both figures almost matches the performance in \Cref{fig:accu_compareuntrusted1}, where the second scheme outperforms the first averaging scheme when attacks are present. However, the convergence becomes slightly slower as the number of SPSs grows. 
In the case of massively distributed data due to the increased number of SPSs, more rounds are required to capture patterns from all available SPSs. 
 \begin{table}[t]
     \centering
      \caption{{Comparative Results for TPRate and TNR of Conventional FedAvg and Proposed Solutions Under Attack and Non-Attack Scenarios (Non-i.i.d. Data Distribution, N-BaIoT Dataset)}.}
     \begin{tabular}{|p{1.3cm}|p{0.5cm}|p{0.4cm}|p{0.5cm}|p{0.5cm}|p{0.7cm}|p{0.5cm}|}
     \hline
         Alg. & \multicolumn{2}{p{0.1cm}} Con. FedAvg & \multicolumn{2}{|p{0.1cm}} P-Trusted &   \multicolumn{2}{|p{0.1cm}} P-Untrusted \\ \hline
         Malicious SPSs & No & Yes &  No & Yes &  No & Yes \\ \hline
        TPRate  & 78.3 & 31.6 & 84.5 & 34.2 & 85.2 & 81.7\\\hline
        TNR & 56.4 & 31.4 & 82.6 & 46.2 & 82.3 & 79.1\\ \hline
        Fairness & 51.2 & 19.7 & 97.1 & 41.2 & 98 & 96.9\\ \hline
     \end{tabular}
     \label{tab:TPR-TNR}
 \end{table}
We also investigate the effect of all proposed schemes, including the reward mechanisms, on the overall performance in terms of detection rate (i.e., TPRate and TNR) and the fairness of reward distribution. From \Cref{tab:TPR-TNR}, it can be shown that proposed averaging methods significantly surpass the conventional FedAvg, whether or not adversarial participants are present. However, when particular participants are adversarial, the proposed trust-based model averaging (P-Trusted) suffers drastically in terms of TPRate, TNR, and reward fairness, which is unsurprising as these are heavily based on the data quality. {In contrast, the proposed untrust-based model averaging successfully filters the contaminated uploaded models using a two-stage verification process}. It distributes the rewards appropriately, achieving a high degree of fairness despite the presence of adversarial participants. FedAvg, on the other hand, distributes its incentives randomly according to the amount of data gathered. Nevertheless, all schemes are affected by the data distribution amongst SPSs, as we note when comparing the i.i.d. to non-i.i.d. results. 
\subsection{{Numerical Results on IEC 104 Dataset}}
 The IEC 104 Dataset was released by Brno University in March 2022 and represents a significant resource in the field of SG security \cite{dataset}. It comprises IEC 104 and IEC MMS headers and is primarily designed for anomaly detection and security monitoring. The dataset captures a broad spectrum of attack scenarios across different folders. For instance, the but-iec104-i folder contains various attacks, including Denial of Service (DoS), Injection, and Man-in-the-Middle (MITM) attacks. These offer a comprehensive view of the threat landscape in the Industrial.
Control Systems (ICS) communications. Another folder, vrt-iec104, presents a different set of intriguing attacks, such as value change and masquerading attacks.{ It is worth mentioning that The IEC 60870-104 (IEC 104) and IEC 61850 (MMS) datasets, produced by the "Security monitoring of communication in the smart grid (Bonnet)" project at the Brno University of Technology, Czech Republic (2019-2022), include CSV traces from PCAP files \cite{dataset}. These datasets, derived from both real device observations and virtual application monitoring, provide a comprehensive view of normal and attack traffic patterns within smart grid environments. This dual approach, blending real-world operational traffic with simulated attack scenarios, ensures the datasets support robust model training and validation, preparing the models to detect and mitigate both existing and emerging threats effectively.}
Despite its richness, the dataset was originally unlabeled and required substantial preprocessing. This step was critical to ensure that our subsequent analyses were based on accurate, well-defined data. The dataset is organized into multiple folders, each containing data related to either IEC 104 or IEC MMS headers. Each folder has a readme.txt file that provides valuable information, including data types and timestamps related to the attacks. This level of detail enabled us to better understand the structure and implications of the data.
We performed preprocessing steps through a multi-stage process to adequately prepare the data for our experimental framework. This involved labeling the samples into categories such as benign, switching, scanning, or communication interruption; cleaning the data to remove any irrelevant or inaccurate information, and finally transforming it into a format suitable for our machine learning algorithms. The encoding was significant for nominal or categorical data. The following section focuses on the numerical analysis conducted using the IEC 104 dataset. We examine performance under two conditions: the i.i.d. and non-i.i.d settings.
 \begin{figure*}[htbp]
    \centering
       \begin{subfigure}[b]{0.32\linewidth}
    \includegraphics[width=\linewidth]{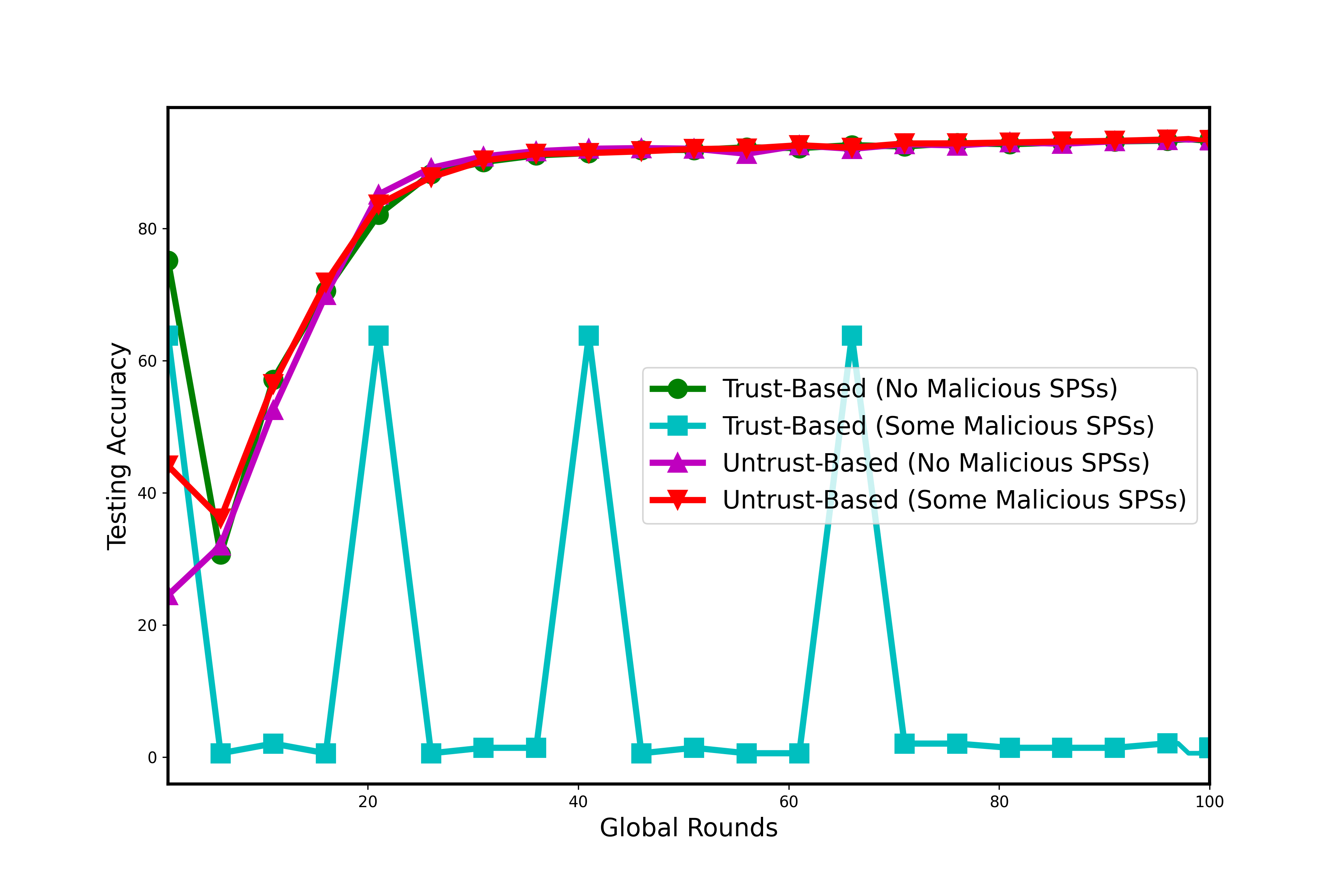}
    \caption{20}
    \label{fig:IEC_accu_iid_20}
     \end{subfigure}
            \begin{subfigure}[b]{0.32\linewidth}
    \includegraphics[width=\linewidth]{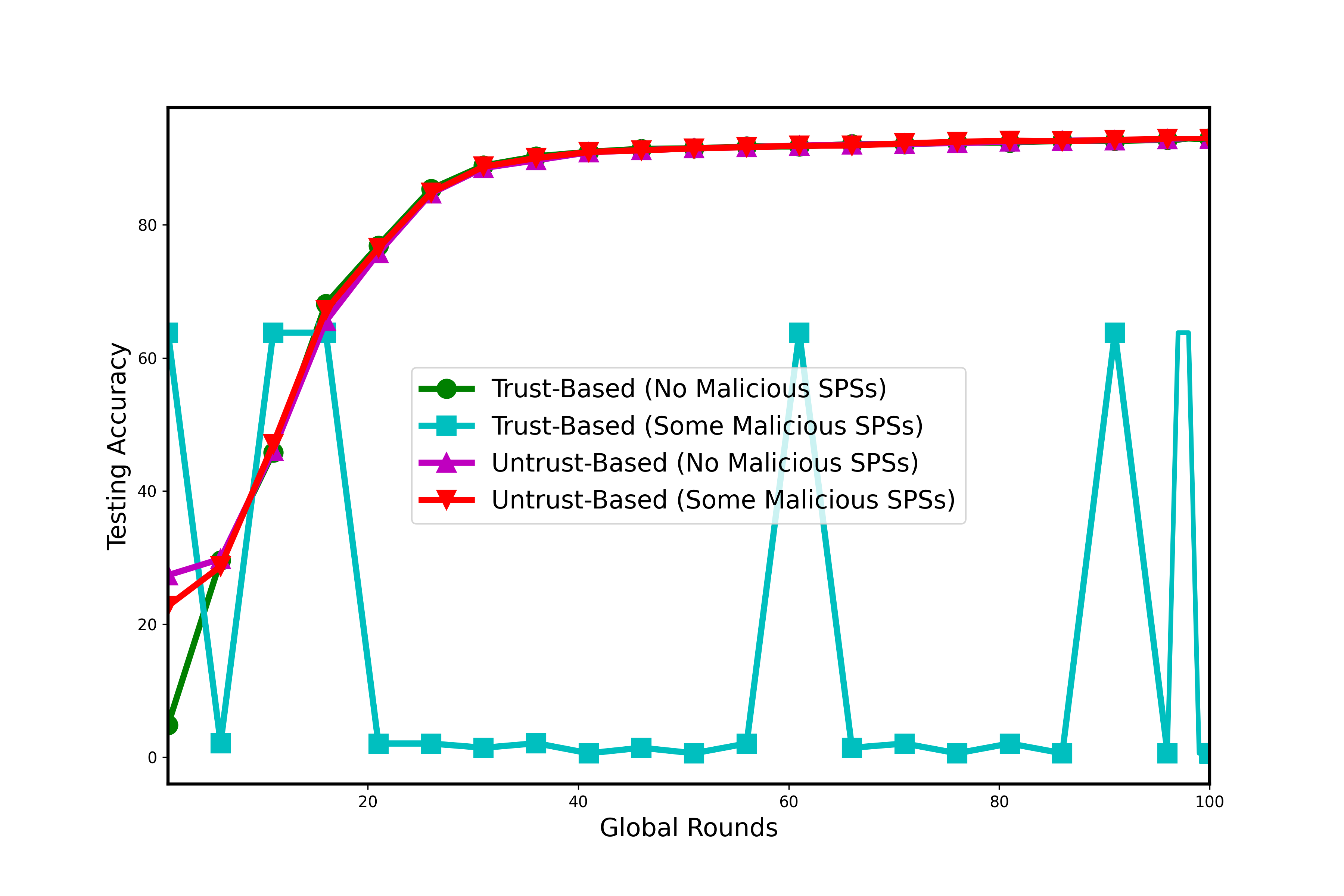}
    \caption{50}
    \label{fig:IEC_accu_iid_50}
     \end{subfigure}
    \begin{subfigure}[b]{0.32\linewidth}
    \includegraphics[width=\linewidth]{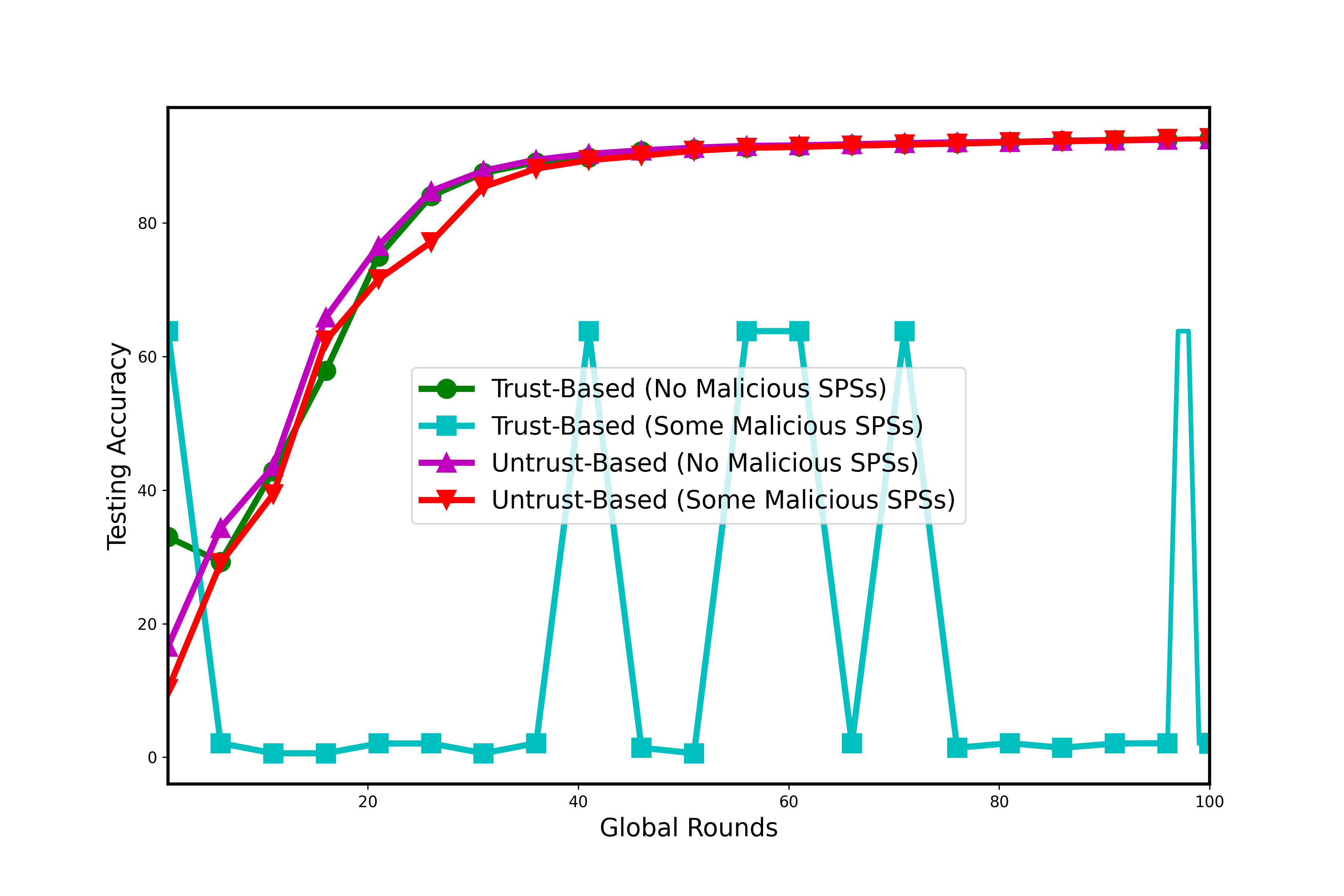}
    \caption{100}
    \label{fig:IEC_accu_iid_100}
     \end{subfigure}
         \caption{{Test Accuracy Comparison of Proposed Schemes in the Presence of Malicious SPSs Claiming High-Quality Data While Uploading Defective Models (I.I.D. Data Distribution, IEC 104 dataset).}}
  \end{figure*}
  
\begin{table}[t]
     \centering
      \caption{{Comparative Results for TPRate and TNR of Conventional FedAvg and Proposed Solutions Under Attack and Non-Attack Scenarios (Non-i.i.d Data Distribution and IEC 104 Dataset).}}
     \begin{tabular}{|p{1.3cm}|p{0.5cm}|p{0.4cm}|p{0.5cm}|p{0.5cm}|p{0.7cm}|p{0.5cm}|}
     \hline
         Alg. & \multicolumn{2}{p{0.1cm}} Con. FedAvg & \multicolumn{2}{|p{0.1cm}} P-Trusted &   \multicolumn{2}{|p{0.1cm}} P-Untrusted \\ \hline
         Malicious SPSs & No & Yes &  No & Yes &  No & Yes \\ \hline
        TPRate  & 81.2 & 28.2 & 86 & 23.8 & 92.7 & 87.3\\\hline
        TNR & 63.1 & 25.1 & 86.2 & 37.4 & 85.3 & 83.6\\ \hline
        Fairness & 69.8 & 14.3 & 98.6 & 32.7 & 97.8 & 94.6\\ \hline
     \end{tabular}
     \label{tab:IEC_TPR-TNR}
 \end{table}
 
 \begin{figure*}[htbp]
    \centering
       \begin{subfigure}[b]{0.32\linewidth}
    \includegraphics[width=\linewidth]{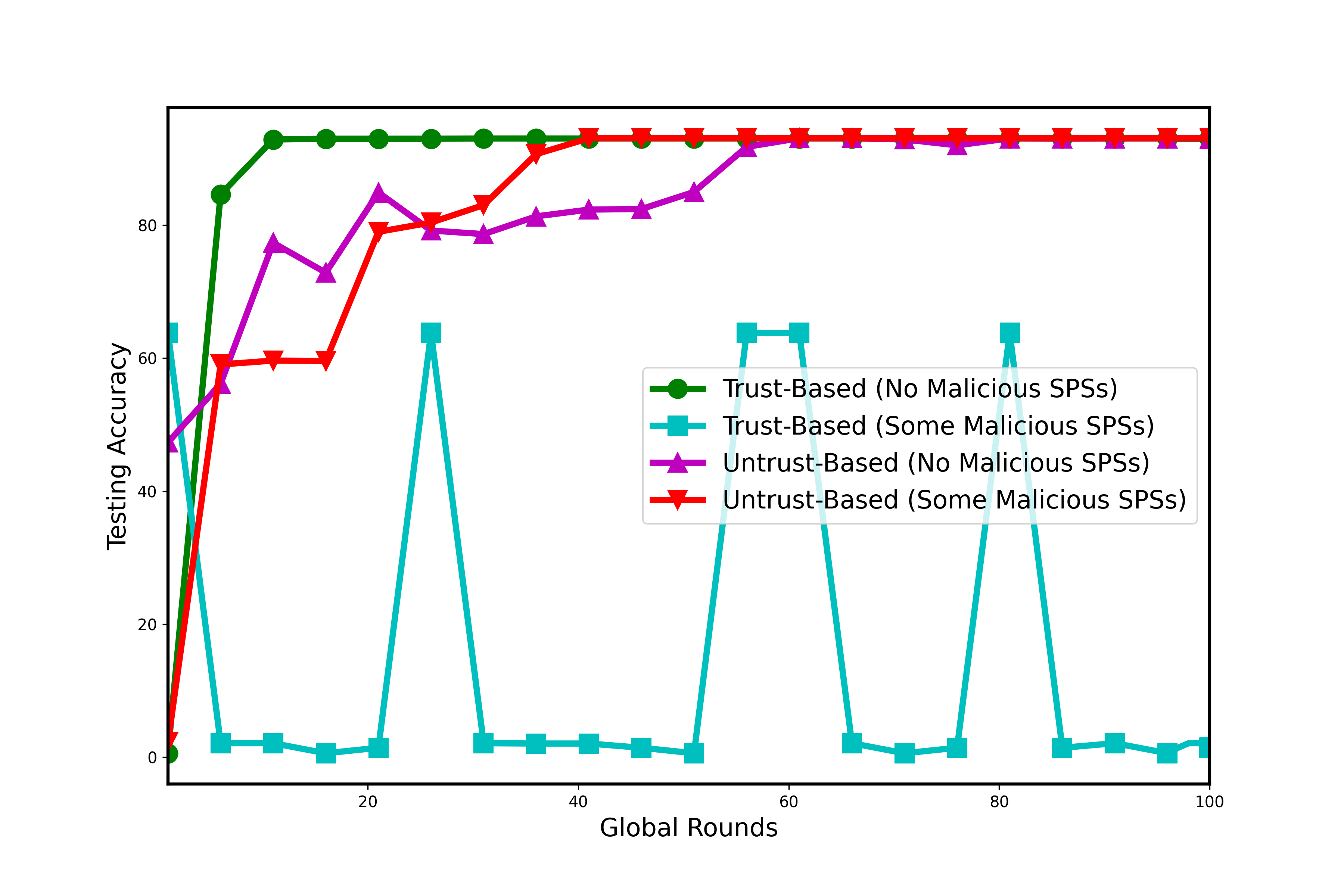}
    \caption{20}
    \label{fig:IEC_accu_noniid_5}
     \end{subfigure}
            \begin{subfigure}[b]{0.32\linewidth}
    \includegraphics[width=\linewidth]{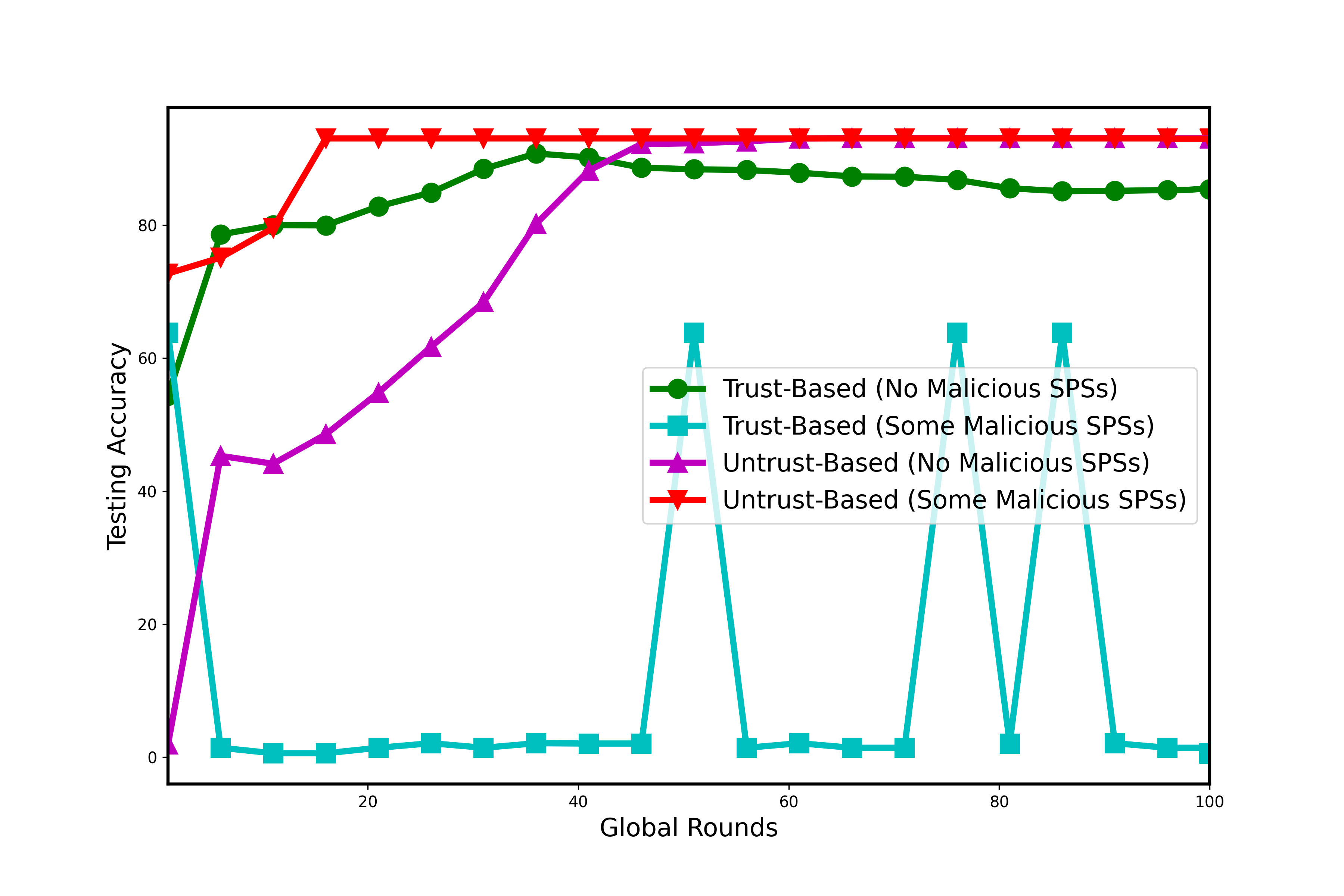}
    \caption{50}
    \label{fig:IEC_accu_noniid_50}
     \end{subfigure}
    \begin{subfigure}[b]{0.32\linewidth}
    \includegraphics[width=\linewidth]{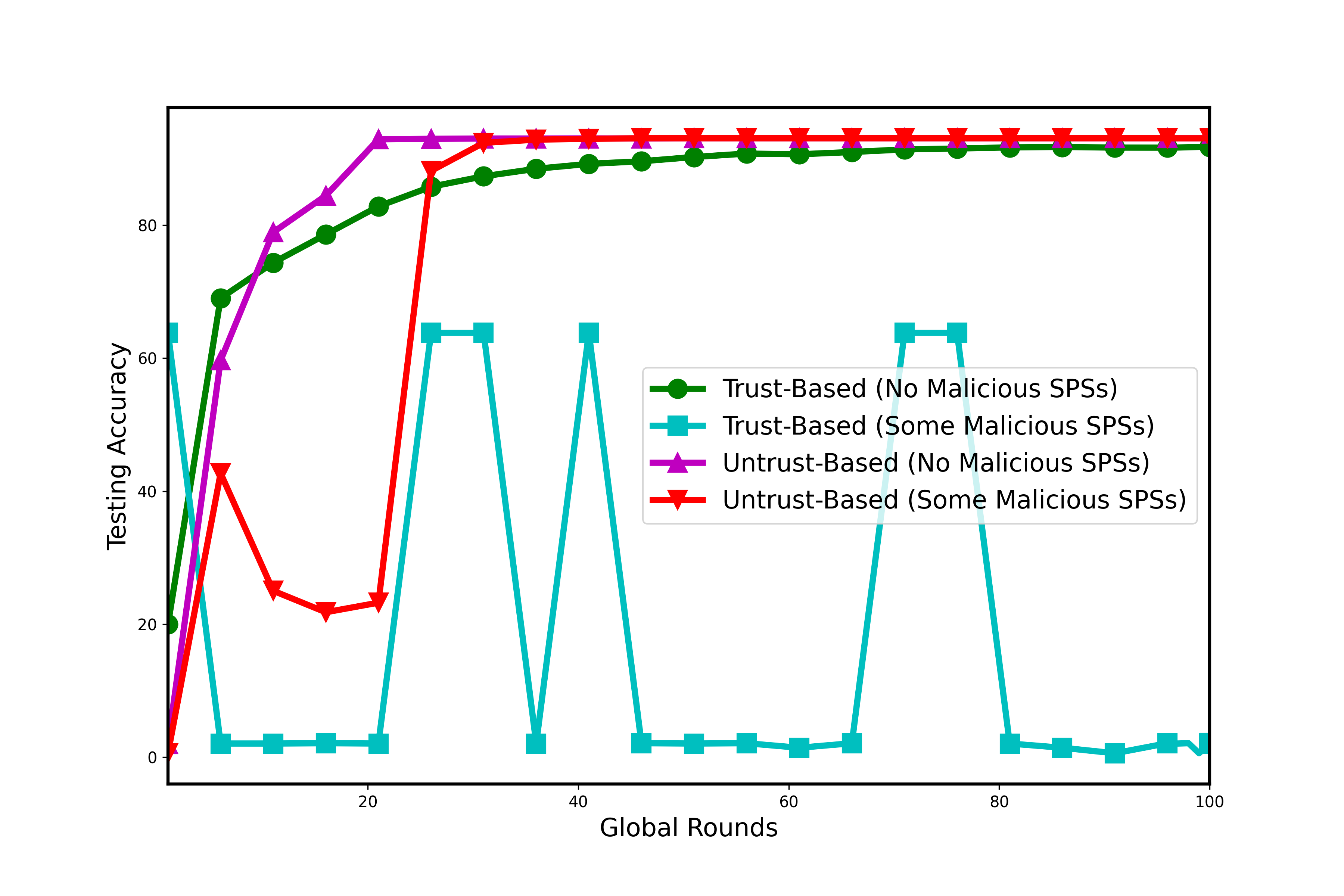}
    \caption{100}
    \label{fig:IEC_accu_noniid_100}
     \end{subfigure}
         \caption{{Test Accuracy Comparison of Proposed Schemes in the Presence of Malicious SPSs Claiming High-Quality Data While Uploading Defective Models (non-I.I.D. Data Distribution, IEC 104 dataset).}}
  \end{figure*}
\subsubsection{{The IID Setting}}
{
We consider three scenarios featuring varying numbers of substations: 20, 50, and 100. \Cref{fig:IEC_accu_iid_20,fig:IEC_accu_iid_50,fig:IEC_accu_iid_100} display the performance of our framework, which includes both trust-based and untrust-based schemes when the data is IID (i.e., all clients have the same data distribution). Our results clearly demonstrate the superiority of the enhanced FedAvg, which achieves performance levels nearly identical to centralized models. This can be attributed to our focus on local data quality for model weight assignment as opposed to merely considering data volume, as in traditional FedAvg. This advantage holds even when malicious participants are involved. Performance significantly drops and fluctuates when only the trust-based scheme is employed. This drop is due to the exclusive reliance on local data quality. However, our untrust-based averaging scheme effectively mitigates this by leveraging generalization test data to allocate weights, neutralizing the impact of corrupted local models.}
\subsubsection{{The Non-IID Setting}}
{In this setting, we assume that each SPS is limited to specific types of IEC 104 traffic patterns. \Cref{fig:IEC_accu_noniid_5,fig:IEC_accu_noniid_50,fig:IEC_accu_noniid_100} show that our proposed untrust-based averaging scheme performs exceptionally well, even in the presence of malicious activity. Trust-based model averaging struggles in this context, experiencing a significant decline in accuracy due to its reliance on claimed data quality.
Furthermore, \Cref{tab:IEC_TPR-TNR} reveals that our approaches outperform traditional FedAvg in critical security metrics such as TPRate and TNR, especially when adversarial participants are present. The trust-based model sees a decline in these metrics due to its dependence on claimed data quality. In contrast, our untrust-based model averaging maintains robust performance, demonstrating its value in real-world scenarios where malicious activity is a concern. Notably, our proposed approaches exhibit superior performance in non-IID settings. This is due to our emphasis on evaluating the quality of returned contributions rather than just data quantity, leading to a more generalizable global model.
In summary, our numerical evaluation of the IEC 104 dataset validates the effectiveness of our proposed methodologies, particularly in non-i.i.d settings and when malicious participants are involved. The insights gained from this analysis will be crucial for refining and optimizing our models for intrusion detection in ICS.}

 \begin{figure*}[htbp]
    \centering
       \begin{subfigure}[b]{0.32\linewidth}
    \includegraphics[width=\linewidth]{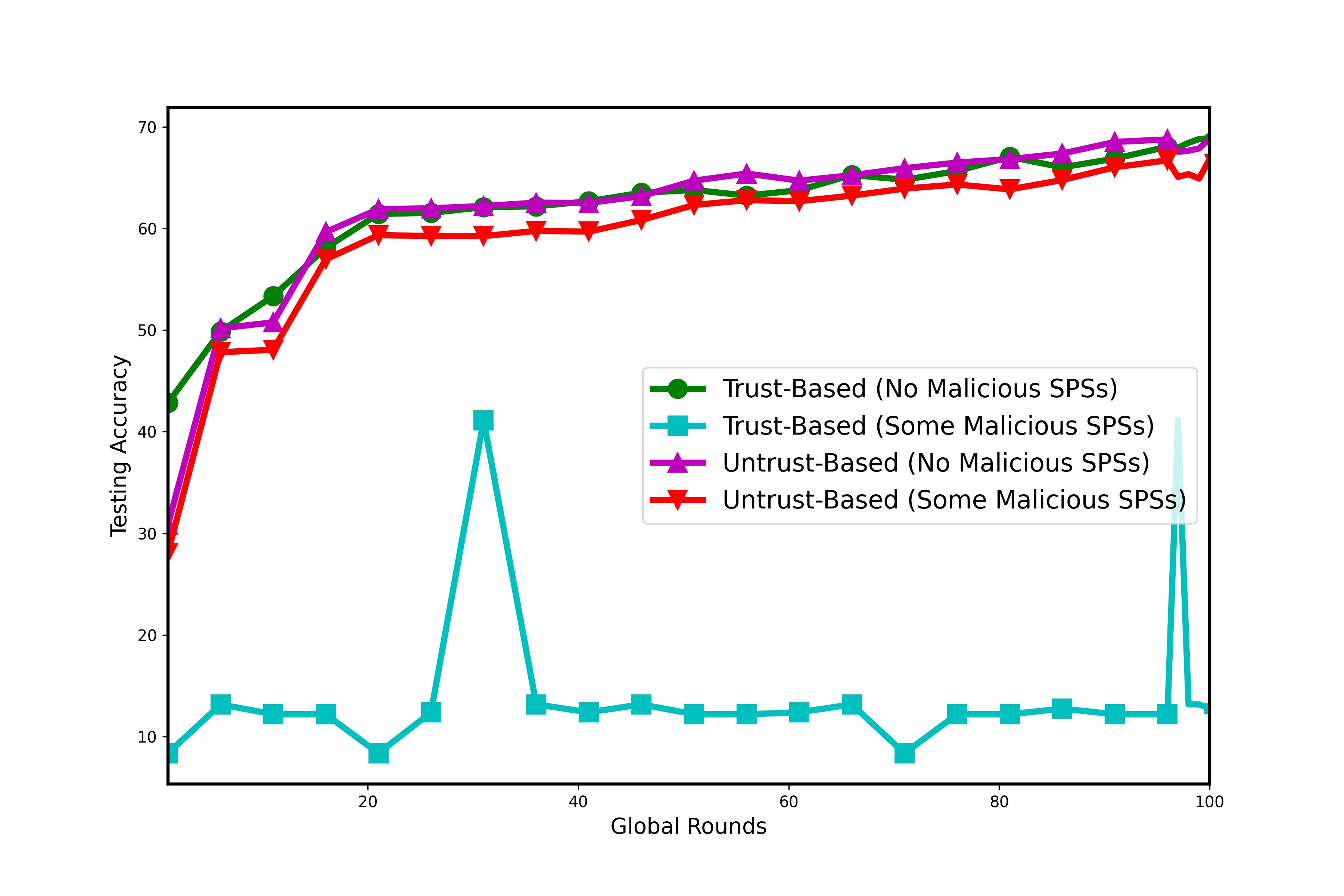}
    \caption{20}
    \label{fig:MMS_accu_noniid_5}
     \end{subfigure}
            \begin{subfigure}[b]{0.32\linewidth}
    \includegraphics[width=\linewidth]{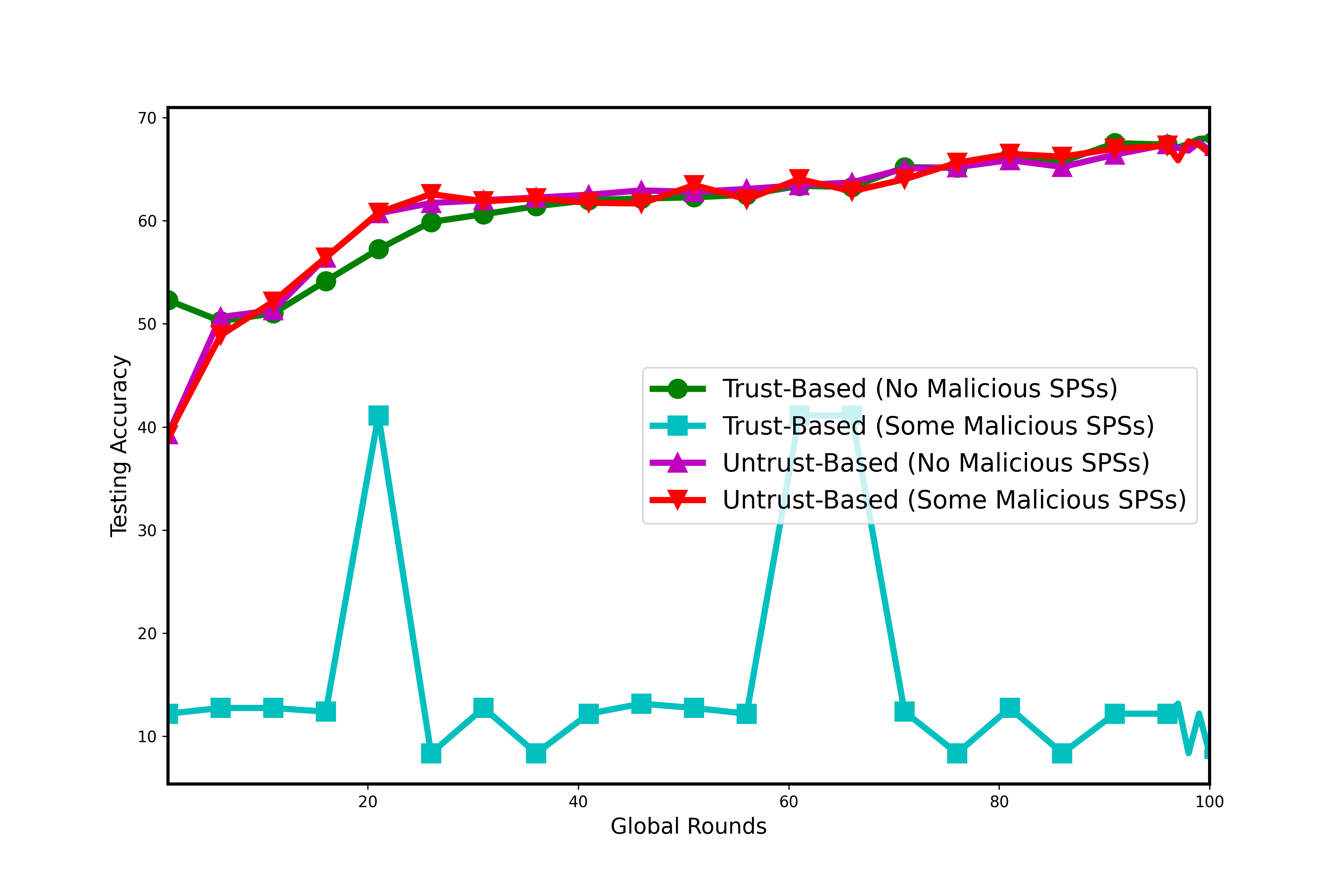}
    \caption{50}
    \label{fig:MMS_accu_noniid_50}
     \end{subfigure}
    \begin{subfigure}[b]{0.32\linewidth}
    \includegraphics[width=\linewidth]{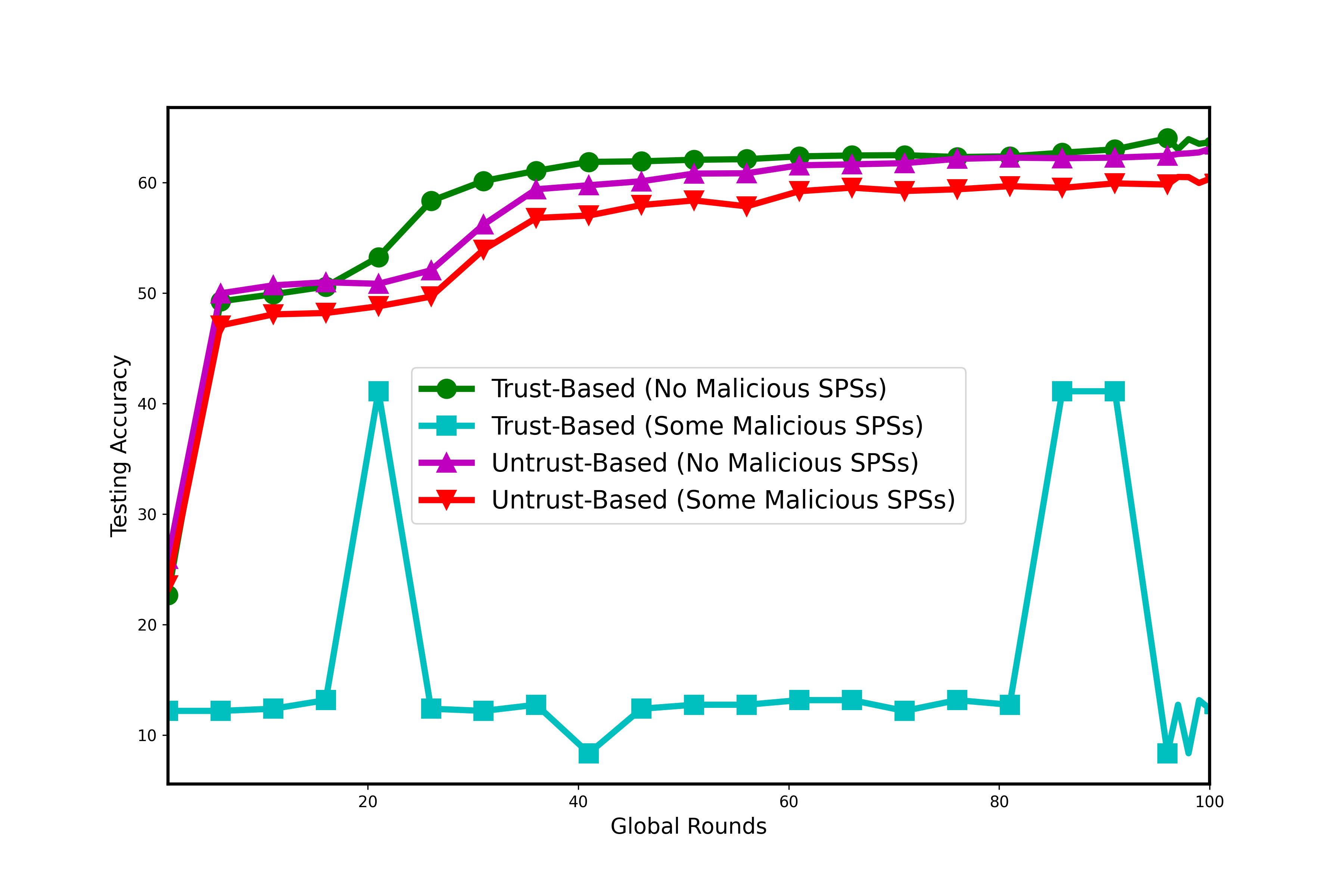}
    \caption{100}
    \label{fig:MMS_accu_noniid_100}
     \end{subfigure}
         \caption{{Test Accuracy Comparison of Proposed Schemes in the Presence of Malicious SPSs Claiming High-Quality Data While Uploading Defective Models (non-I.I.D. Data Distribution, IEC 61850 (MMS) dataset).}}
  \end{figure*}
\subsection{{Verfication with  IEC 61850 Manufacturing Message Specification (MMS) dataset}}
{To further verify our results, we carry out experiments using the IEC 61850 (MMS) dataset. IEC 61850 (MMS) is extensively employed in electric utility companies, predominantly for substation automation and ensuring interoperability between different manufacturers' systems. This global standard is important for streamlining the communication infrastructure of electrical substations and facilitating the integration, operation, and maintenance of diverse devices and systems in SG. Its usage spans various applications, including real-time monitoring, control of substation components, and ensuring robust, reliable data exchange. The dataset used includes both benign and malicious data samples. Remarkably, the performance trends observed with this dataset closely mirror those identified using the IEC 104 protocol data. Specifically, under non-IID settings as a most challenging scenario, our enhanced FedAvg framework, which includes both trust-based and untrust-based schemes, continued to demonstrate superior performance, closely approximating that of centralized models. This consistency highlights the robustness of our approach, particularly our novel method of weighting model updates based on the quality of local data rather than its volume.  Including the IEC 61850 dataset not only verifies our initial findings but also broadens the applicability of our framework. It shows that our untrust-based averaging scheme, which assesses the quality of contributions through generalization on test data, can effectively neutralize the influence of corrupted local models, ensuring stable and reliable model performance even in the face of malicious activity. This result demonstrates the generalizability and robustness of our approach across different SG communication protocols, further validating the effectiveness of our method in enhancing security and reliability in real-world, heterogeneous SG network environments.}

\section{Conclusion}
\label{sec:conclusion}
In this paper, we introduced FedPot, a novel quality assurance honeypot-based FL framework designed for network security in SG. FedPot incorporates {a} novel, efficient, and resilient aggregation and averaging schemes coupled with a fair rewards mechanism.  We presented novel schemes for local data quality, participant selection, and global model upgrading using the N-BaIoT, IEC 104, and IEC MMS datasets. In FedPot, the TPR addresses a convex optimization problem, prioritizing data quality over data size. Each SPS optimizes the global model with its honeypot logs and transmits the model updates back to the TPR. Subsequently, the TPR enhances the defensive model using the approaches proposed in this study.
To mitigate the free-rider issue prevalent in AMI networks within the FL framework, we proposed a new metric to gauge local data quality and contributions, {eliminating} the need to rely on data size. We also devised a two-step verification process to address the challenge of adversaries or underperforming SPSs. {Additionally}, we introduced an improved FedAvg scheme for local model aggregations.
The results obtained from extensive simulations with realistic log data attest to the effectiveness of our proposed scheme, which outperforms current state-of-the-art techniques. {As a direction} for future {research}, investigating the real-time implementation and assessment of FedPot in streaming, diverse, and larger-scale environments would be insightful. Lastly, adapting the FedPot framework to address other cybersecurity threats within various IIoT applications could also be advantageous.

\section*{Acknowledgement}
This publication was made possible by NPRP Cluster project (NPRP-C) Twelve (12th) Cycle grant \# NPRP12C-33905-SP-67 from the Qatar National Research Fund (a member of Qatar Foundation). The findings herein reflect the work, and are solely the responsibility, of the authors.
\bibliographystyle{IEEEtran}
\bibliography{ref}

\begin{IEEEbiography} [{\includegraphics[width=1in,height=1.25in,clip,keepaspectratio]{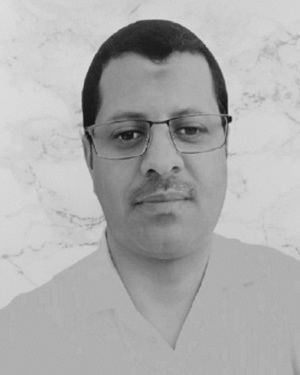}}] { \rmfamily  Abdullatif Albaseer (Member, IEEE)} \rmfamily \mdseries received an M.Sc. degree in computer networks from King Fahd University of Petroleum and Minerals, Dhahran, Saudi Arabia, in 2017 and a Ph.D. degree in computer science and engineering from Hamad Bin Khalifa University, Doha, Qatar, in 2022. He is a Postdoctoral Research Fellow with the Smart Cities and IoT Lab at Hamad Bin Khalifa University. He has authored and co-authored over thirty conference and journal papers in IEEE ICC, IEEE Globecom, IEEE CCNC, IEEE WCNC, and IEEE Transactions. He also has six US patents in the area of the wireless network edge. His current research interests include AI for Networking, AI for Cybersecurity, Distributed AI, and Edge LLMs.
\end{IEEEbiography}

\begin{IEEEbiography} [{\includegraphics[width=1in,height=1.25in,clip,keepaspectratio]{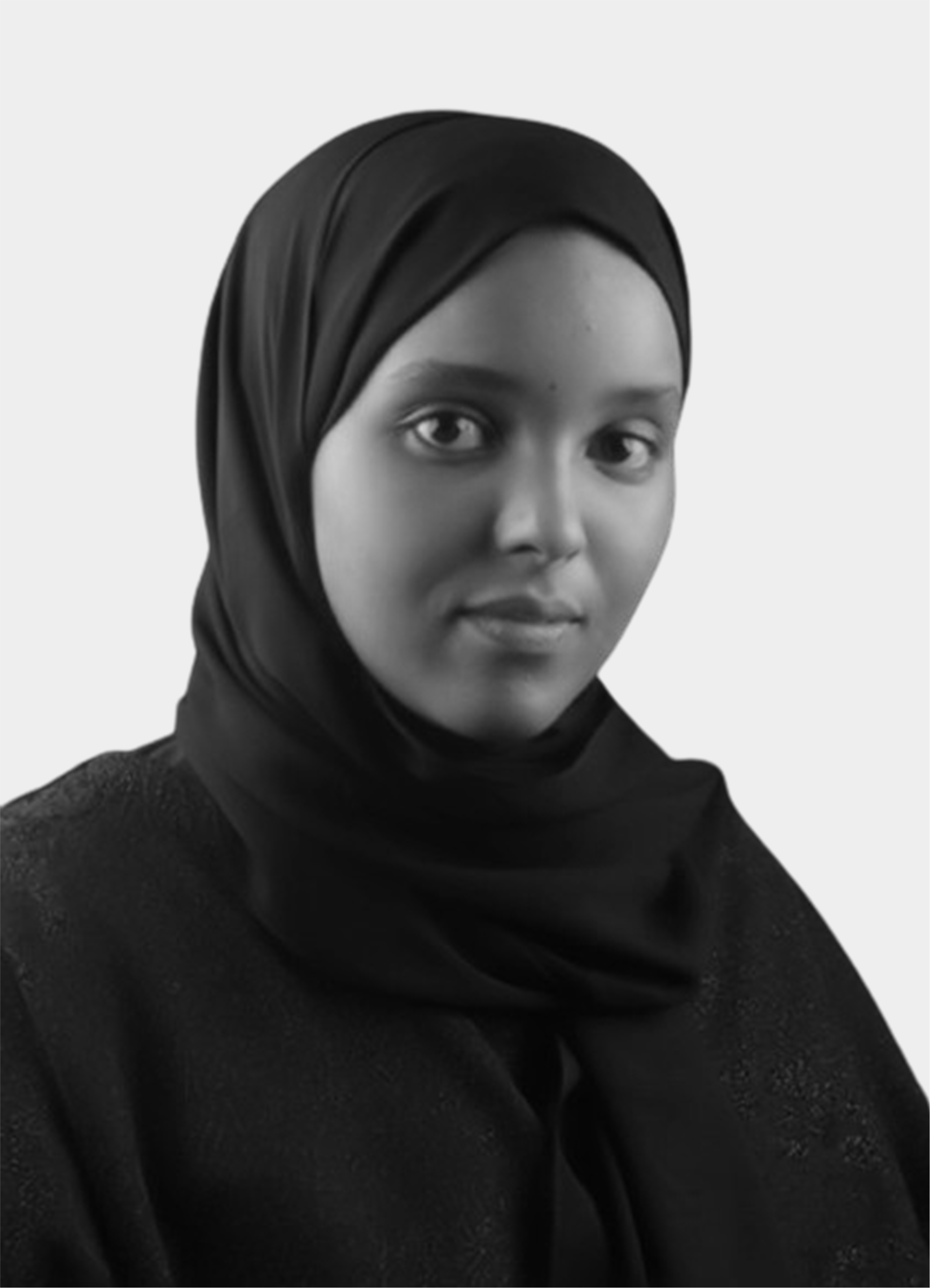}}]  { \rmfamily  Nima Abdi}
\rmfamily \mdseries received her B.Sc. in Electrical Engineering from Qatar University in 2020 and is currently pursuing an M.Sc. in Data Science and Engineering at Hamad Bin Khalifa University (HBKU). Her research focus is on the application of Artificial Intelligence on smart grid security, specifically the physical layer.
\end{IEEEbiography}
\begin{IEEEbiography} [{\includegraphics[width=1in,height=1.25in,clip,keepaspectratio]{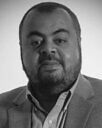}}]
{ \rmfamily  Mohamed Abdallah (Senior Member, IEEE)} \rmfamily \mdseries
received his B.Sc. degree from Cairo University, Giza, Egypt, in 1996, and his M.Sc. and Ph.D. degrees from the University of Maryland at College Park, College Park, MD, USA, in 2001 and 2006, respectively.,From 2006 to 2016, he held academic and research positions with Cairo University and Texas A \& M University in Qatar, Doha, Qatar. He is currently a Founding Faculty Member with the rank of Associate Professor with the College of Science and Engineering, Hamad Bin Khalifa University, Doha. He has published more than 150 journals and conferences and four book chapters and co-invented four patents. His current research interests include wireless networks, wireless security, smart grids, optical wireless communication, and blockchain applications for emerging networks. He is a recipient of the Research Fellow Excellence Award at Texas A\& M University in Qatar in 2016, the Best Paper Award in multiple IEEE conferences, including IEEE BlackSeaCom 2019 and the IEEE First Workshop on Smart Grid and Renewable Energy in 2015, and the Nortel Networks Industrial Fellowship for five consecutive years, 1999–2003. His professional activities include an Associate Editor of the IEEE Transactions on Communications and the IEEE Open Access Journal of Communications, the Track Co-Chair of the IEEE VTC Fall 2019 Conference, the Technical Program Chair of the 10th International Conference on Cognitive Radio-Oriented Wireless Networks, and a technical program committee member of several major IEEE conferences.
\end{IEEEbiography}

\begin{IEEEbiography} [{\includegraphics[width=1in,height=1.25in,clip,keepaspectratio]{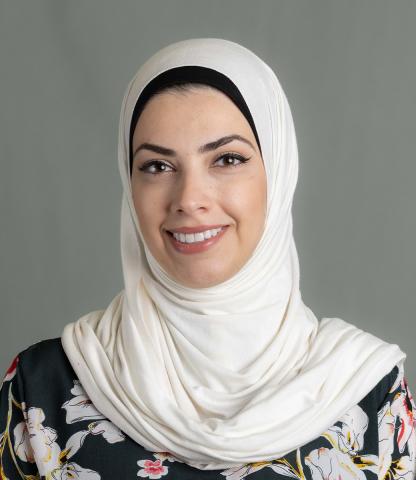}}]
{ \rmfamily  Marwa Qaraqe (Senior Member, IEEE)} \rmfamily \mdseries
 is an Associate Professor within the Division of Information and Communication Technology at Hamad Bin Khalifa University's College of Science and Engineering. She completed her bachelor’s degree in Electrical Engineering at Texas A\&M University in Qatar in 2010, followed by her MSc and PhD in Electrical Engineering at Texas A\&M University in College Station, TX, USA, in August 2012 and May 2016, respectively. Dr. Qaraqe's research focuses on various aspects of wireless communication, signal processing, and machine learning, with applications spanning multidisciplinary areas such as security, IoT, and health. Her specific interests lie in physical layer security, federated learning across wireless networks, and employing machine learning techniques for enhancing wireless communication, security, and healthcare systems. She has been actively involved in developing physical layer security protocols for IoT networks and has secured a NATO SPS grant for her work in this domain. Additionally, Dr. Qaraqe is engaged in research exploring emerging technologies like RIS (Reconfigurable Intelligent Surfaces) and reinforcement learning to advance the capabilities of wireless communication, particularly in the context of enabling efficient and highly secure communication infrastructures for smart cities.
\end{IEEEbiography}

\begin{IEEEbiography} [{\includegraphics[width=1in,height=1.25in,clip,keepaspectratio]{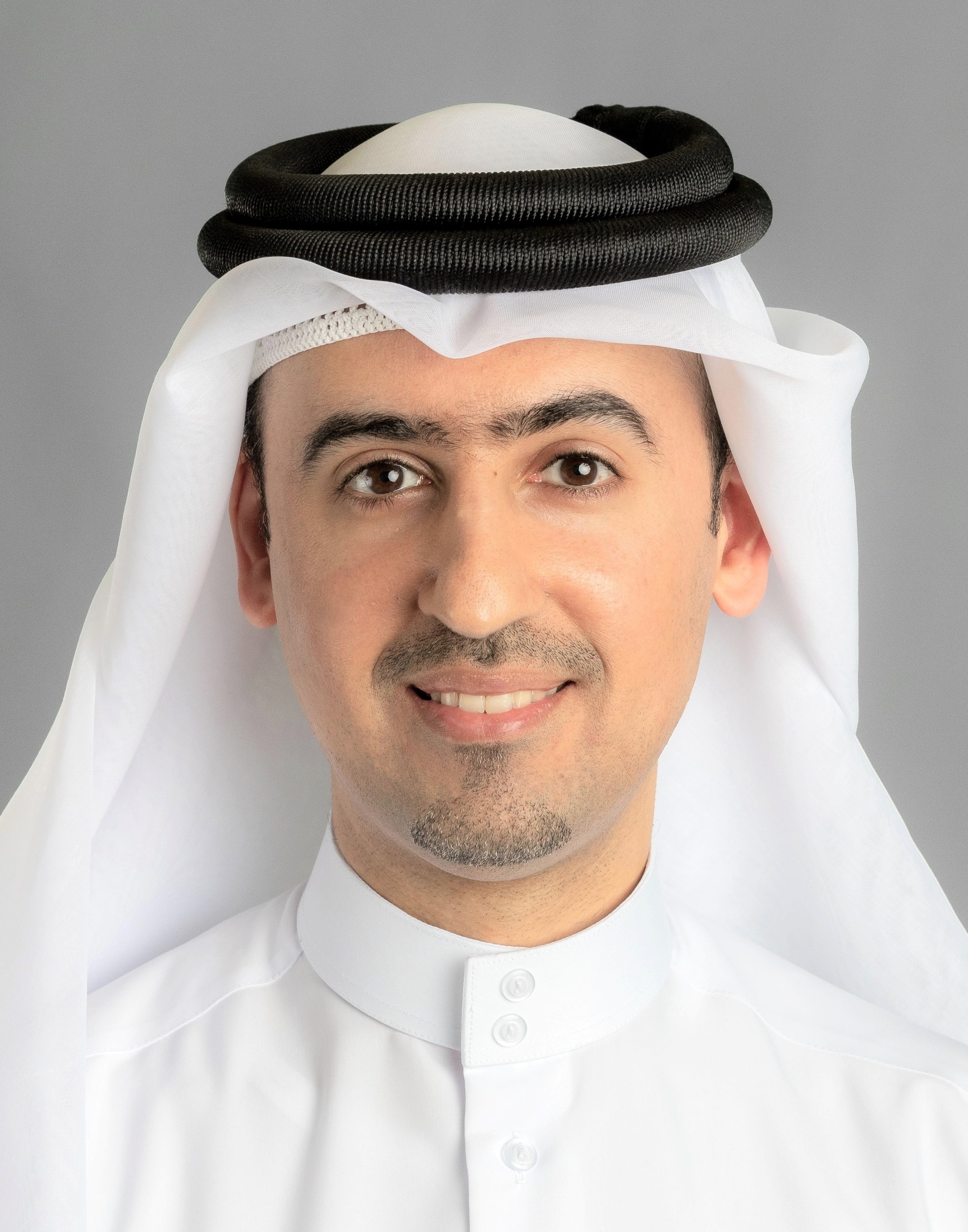}}]
{ \rmfamily  Saif Al-Kuwari (Senior Member, IEEE)} \rmfamily \mdseries
 received a Bachelor of Engineering in Computers and Networks from the University of Essex (UK) in 2006 and two PhD’s from the University of Bath and Royal Holloway, University of London (UK) in Computer Science, both in 2011. He is currently a faculty at the College of Science and Engineering at Hamad Bin Khalifa University and the director of the Qatar Center for Quantum Computing (QC2). His current research interests include, mainly, quantum cryptography and quantum machine learning. He is IET and BCS fellow, and IEEE and ACM senior member.
 
\end{IEEEbiography}
\end{document}